\documentclass[pdflatex,iicol]{sn-jnl}
\usepackage[numbers,sort&compress]{natbib}
\usepackage[T1]{fontenc}
\usepackage{graphicx}
\usepackage{xcolor,manyfoot}
\DeclareNewFootnote{A}[gobble]
\usepackage{amsmath,amssymb,bm}
\usepackage{amsthm}
\theoremstyle{thmstyleone}
\newtheorem{theorem}{Theorem}
\hypersetup{hidelinks,pdfauthor={Hongyu Ma},
pdftitle={Information-preserving boundary reconstruction for center-flux energetics in a regulated SU(3) gauge ladder}}

\makeatletter
\newcommand{\displayappendixsection}[1]{%
  \let\appendixsavedhangfrom\@hangfrom
  \def\@hangfrom##1{\noindent##1\par\nobreak\noindent}%
  \section{#1}%
  \let\@hangfrom\appendixsavedhangfrom
}
\makeatother

\begin{document}
\title{Information-preserving boundary reconstruction for center-flux energetics
in a regulated SU(3) gauge ladder}
\author*{\fnm{Hongyu} \sur{Ma}}\email{hongyu.ma@alu.scu.edu.cn}
\affil{Independent Researcher, Changzhou, China}
\abstract{Gauss constraints make it nontrivial to open, close or join a gauge
system while retaining its interior physics. Boundary reconstructions are
constructed on a pure SU(3) ladder at six site-singlet cutoffs, with all allowed
intertwiner multiplicities. Replacements of prescribed length are obtained using
legal column paths and physical self-loops.
On their stated physical domains, the channels preserve unnormalized conditional
interior matrices, including their weights and within-branch
coherence. Lowest-cutoff dephasing controls show that path
probabilities alone do not preserve the matched magnetic
observables. In that example, coherent resources lower the output energy
without purification or global ground-state preparation.
Matched interior Hamiltonian terms cancel exactly, so the energy cost is bounded
independently of longitudinal length.
The vacuum and flux energy densities are then established separately by
same-sector gluing and stability. The open and periodic vacuum-subtracted
line-energy densities are shown to be equal by two same-length boundary
comparisons using the qualified periodic minimizing domain.
These limits hold at fixed lattice spacing, transverse width,
coupling and chosen cutoff.}
\maketitle

\section{Introduction}
\label{sec:introduction}

Electric-flux sectors describe how a pure non-Abelian gauge field
stores energy without dynamical fundamental
matter~\cite{tHooft1979,Yaffe1980ConfinementSUN}.
Their energetics has long been studied in Hamiltonian SU(3) lattice
gauge theory~\cite{KogutPearsonShigemitsu1981,KogutPearsonShigemitsu1979Beta}.
For a regulated ladder, two questions are closely related:
does the vacuum-subtracted flux energy have a longitudinal density,
and does that density depend on whether the ladder is open or
periodic? Finite spectra alone cannot settle either question.
A controlled comparison of physical states at different lengths and boundaries
is required to settle these questions.

The center-charge setting is also studied through one-form
symmetry~\cite{GaiottoEtAl2015GeneralizedSymmetries,HsinLamSeiberg2019OneForm,%
AddII04} and disorder or flux-sector observables~\cite{AddVII03,AddII20}.
Strong-coupling and Hamiltonian studies provide the historical
setting~\cite{KogutSinclairSusskind1976LowEnergyQCD,AddVI28}.
Confinement benchmarks include glueball and string-tension
calculations~\cite{AddVI09,AddVI16,AddVI36}.
Finite-volume context comes from electric-flux and twisted-sector
studies~\cite{AddII09,KollerVanBaal1986,Luscher1983,GarciaPerezEtAl2018,%
AddII21,VanBaal1982Hypertorus}
and related torus spectra and volume
dependence~\cite{AddII11,AddII22,AddII05,AddII15,AddVI18}.

Gauss constraints obstruct ordinary tensor-product cutting and
joining~\cite{DelcampDittrichRiello2016NonAbelianEntanglement,%
SoniTrivedi2016GaugeEntanglement,Donnelly2014NonAbelianEntanglement,%
AokiEtAl2015GaugeEntanglementDefinition,%
BuividovichPolikarpov2008GaugeEntanglement,LinRadicevic2020Entanglement}. The
representation data left by a removed piece must be matched by its replacement,
as also occurs in gauge-theory
subregions~\cite{CasiniHuertaRosabal2014,Donnelly2012BoundaryRepresentations,%
HateganMarandiuc2024NonAbelianEntanglement,DonnellyFreidel2016LocalSubsystems}.
The replacement must also have the right length: an open ladder
of $N$ plaquettes has $N+1$ complete columns, while a periodic one
has $N$. Both requirements are satisfied using legal column paths and physical
self-loops. The construction is then extended from basis paths to quantum states
by conditioning on the exposed rail data.
Related boundary descriptions use flux and gluing
variables~\cite{GomesRiello2021QuasilocalYM,%
RielloSchiavina2024FluxSuperselection,Riello2021YMFluxSuperselection,%
GomesHopfmullerRiello2019BoundaryCharges}, or superselection sectors and edge
modes~\cite{BallCiambelli2026DynamicalEdgeModes,%
FeldmanEtAl2024SuperselectionEntanglement,VanAcoleyenEtAl2016DistillationGauge,%
GeillerJaiakson2020EdgeModes}.

The resulting channel preserves unnormalized conditional interior
matrices, including their original weights and within-branch
coherence. In the lowest-cutoff example, dephasing controls show
why path probabilities alone do not suffice: interior magnetic
observables depend on coherence that diagonal electric data
cannot detect. The boundary state can still be chosen while this interior
information is preserved. Coherent legal resources can lower the output energy
without purifying the output state or preparing a global ground state.

The same information preservation is used in the energy comparison.
Matched interior Hamiltonian terms cancel exactly, so only a bounded collection
of boundary terms remains. Same-sector gluing is allowed by this
length-independent cost, and each sector density is then established using
stability and a tail subadditivity argument.
The sector densities for open and periodic realizations are then shown to be
equal by two independent maps at the same length. Vacuum subtraction
is taken only after the sector limits exist.
These results are stated in the two theorems below. Resource optimization and
finite spectra are not premises of either theorem.

The construction uses six site-singlet cutoffs with all retained
intertwiner multiplicities. The longitudinal limit is taken with the lattice
spacing, transverse width, coupling and chosen cutoff held fixed.
Section~\ref{sec:longitudinal} states the scope of
the resulting center-flux line-energy density.

\section{The SU(3) ladder and its physical path spaces}
\label{sec:physical-paths}

A pure SU(3) plaquette ladder without dynamical fundamental matter is
considered.
The global form of the gauge group specifies the allowed line
operators~\cite{AharonySeibergTachikawa2013}. The lattice spacing $a$, finite
transverse extent, coupling
convention and representation cutoff are fixed as the longitudinal
length $L=Na$ varies. Numerical energies use $a=1$ and
\begingroup
\setlength{\abovedisplayskip}{6pt plus 2pt minus 1pt}
\setlength{\belowdisplayskip}{6pt plus 2pt minus 1pt}
\begin{equation}
\label{eq:numerical-hamiltonian}
 \begin{aligned}
 H_{\rm num}&=\frac{g^2}{2}\sum_\ell C_2(r_\ell)
 -\frac1{g^2}\sum_{j=1}^N(\Box_j+\Box_j^\dagger),\\
 &\qquad g^2>0.
\end{aligned}
\end{equation}
\endgroup
Here $N$ counts plaquettes and $r_\ell$ is the irrep on a physical
link. The electric term is diagonal~\cite{KogutSusskind1975}.
The oriented plaquette contracts the four local singlet tensors with
surrounding rail data fixed; the same contraction and its adjoint
enter once at every length. A common state-independent magnetic
constant has been removed.

Gauss's law is imposed in the basis~\cite{CiavarellaKlcoSavage2021}.
Related Hamiltonian
bases~\cite{BurgioEtAl2000PhysicalHilbert,AnishettyMathurRaychowdhury2010,%
LigterinkWaletBishop2000ManyBodySUN,BronzanVaughan1991HamiltonianQCDBasis} and
treatments of Gauss
constraints~\cite{MathewRaychowdhury2025ProtectingSU3,%
PardoEtAl2023GaussFermionElimination} provide
alternative constructions.
At a trivalent vertex, the label $\gamma$ distinguishes orthonormal
singlets in $R_t\otimes\overline{R_b}\otimes
R_v$~\cite{deSwart1963SU3CG,Kaeding1995SU3Isoscalar,AddIII22,%
Hecht1965SU3Recoupling}. Every such singlet
is retained when
\begin{equation}
\label{eq:site-cutoff}
 \begin{aligned}
 & C_2(R_t)+C_2(R_b)+C_2(R_v)\leq B,\\
 & B\in\left\{4,\frac{17}{3},6,\frac{20}{3},\frac{23}{3},9\right\}.
\end{aligned}
\end{equation}
The cutoff is on the incident sum, not on links
independently~\cite{BalajiEtAl2026}.
In particular, the two $8\otimes8\otimes8$ singlets at $B=9$
remain distinct.
SU(3) coupling and singlet constructions retain multiplicity
information~\cite{AddIII16,AddIII15,AnishettySreeraj2019Singlets}.
Outer multiplicities also enter trivalent-vertex
bases~\cite{KadamEtAl2025} and the general coupling
problem~\cite{PanDraayer1998OuterMultiplicityI,AddIII04,AddIII17,%
PanDraayer1998OuterMultiplicityII}.

Let $(U_i,D_i)$ be the upper and lower rail representations at a
cut and $V_i$ the vertical representation in column $i$.
The upper and lower singlets belong to
\begin{equation}
\label{eq:column-singlets}
 U_i\otimes\overline{U_{i-1}}\otimes V_i,
 \qquad
 D_{i-1}\otimes\overline{D_i}\otimes\overline{V_i}.
\end{equation}
A complete column is therefore the directed change
$(U_{i-1},D_{i-1})\to(U_i,D_i)$ together with
$(V_i,\gamma_t,\gamma_b)$. Columns with the same endpoints and
different witnesses are different basis states.
They form a finite directed multigraph $\mathcal G_{B,q}$ at
conserved center charge
\begin{equation}
\label{eq:center-charge}
 \begin{aligned}
 & q=t(D)-t(U)\pmod3,\\
 & t(R)=p+2q_{\rm D}\pmod3,
\end{aligned}
\end{equation}
where $(p,q_{\rm D})$ are Dynkin labels. The vacuum is $q=0$.
Fundamental Wilson winding along $-z$ selects the flux convention
$q=1$; opposite winding gives $q=2$.

For the canonical rail pairs
\begin{equation}
\label{eq:canonical-pairs}
 c_0=(1,1),\qquad c_1=(\bar3,1),\qquad c_2=(3,1),
\end{equation}
the open physical space $\mathcal H_{N,B}^{(q,\mathrm o)}$ is
spanned by complete $(N+1)$-edge paths from $c_q$ to $c_q$.
The periodic space $\mathcal H_{N,B}^{(q,\mathrm p)}$ is spanned
by complete labelled $N$-edge closed walks. These are orthonormal
path bases with all multiplicities; periodic positions are not
identified under translations. Both describe $N$ plaquettes.
The graph records legal concatenation, not Hamiltonian transitions.
The full periodic space includes every active component, a
distinction used in Sec.~\ref{subsec:general-domains}.

The difference of one column is central to changing boundaries at
fixed length. After an open path is cut, the replacement must match
its exposed representations and the required periodic column
count. Reachability of the endpoints alone does not ensure the latter.

Alternative formulations use loop or Schwinger-boson
variables~\cite{AddI03,AddIII20,AddIII21,AddIII08,AddIII28}, or holonomies,
duality and gauge fixing~\cite{AddIV11,AddIII02,AddIII19,AddIV14,AddIII18}.
Other choices include monomer-dimer and q-deformed
constructions~\cite{AddI01,HayataHidaka2023QDeformed,AddIII25}, and improved or
resource-oriented Hamiltonian
formulations~\cite{AddIV07,AddIII13,ZoharBurrello2015}.
Further examples use spin networks and q
deformation~\cite{AddIII06,AddIV09,AddVIII02}, improved honeycomb
Hamiltonians~\cite{AddIII14}, or quantum links and discrete gauge
groups~\cite{BrowerChandrasekharanWiese1999,AddIV08,AddII06}.

Truncation strategies and their errors are studied in
Refs.~\cite{AddI04,AddI05,LiuChandrasekharan2022QubitRegularization,AddIV02}.
Digitization and Hamiltonian
improvement~\cite{AddIV16,AddIV15,CarenaEtAl2022ImprovedHamiltonians}
offer further comparisons, with simulation bounds given in
Ref.~\cite{TongEtAl2022ProvablyAccurate}.

\section{A minimal boundary reconstruction}
\label{sec:minimal-reconstruction}

\subsection{A legal bridge at the same physical length}
\label{subsec:b4-exact-bridge}

At $B=4$ and $q=1$, the active rail pairs are
\begin{equation}
\label{eq:b4-vertices}
 x=(1,3),\qquad y=(3,\bar3),\qquad c_1=(\bar3,1).
\end{equation}
Every ordered pair has one complete column, including equal
endpoints, with $\gamma_t=\gamma_b=0$.
Suppose a retained segment from $y$ to $x$ needs a five-column
closing bridge. A direct edge $x\to y$ is too short.
Instead use $e_{xc}:x\to c_1$ and $e_{cy}:c_1\to y$, both with
$V=3$, and the physical self-loop $e_c:c_1\to c_1$ with $V=1$.
The loop contains the singlets
$\bar3\otimes3\otimes1$ and $1\otimes1\otimes1$, within the
site cutoff. Repeating it gives
\begin{equation}
\label{eq:five-edge-bridge}
 w_{xy}=e_{xc}\,e_c\,e_c\,e_c\,e_{cy}.
\end{equation}
Each traversal is a different column position.
The loop therefore supplies length without changing the exposed
rails or charge. Figure~\ref{fig:exact-bridge} shows how this
fixed bridge closes a retained edge in the $N=6$ example.
Its choice is geometric, not an energy minimization.

\begin{figure*}[tbp]
\centering
\includegraphics[width=0.5333333333\textwidth]{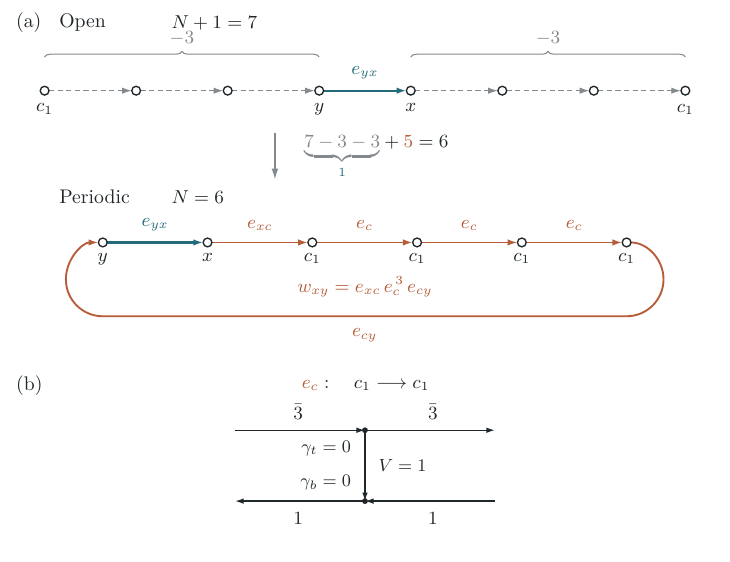}
\caption{Exact-length boundary reconstruction at $B=4$ and $q=1$.
(a) Removing three columns (gray dashed) at each end of an open
seven-column $N=6$ basis path retains $e_{yx}$ (teal).
The fixed five-column bridge $w_{xy}$ (rust) closes it into a
six-column periodic walk, preserving the six plaquettes and physical
length $6a$. Circles denote the rail pairs of
Eq.~\eqref{eq:b4-vertices}. Each directed edge denotes one complete
column, including the curved return $e_{cy}$. Repeated $c_1$ labels
mark successive visits to the same rail-pair state.
(b) The physical column for $e_c:c_1\to c_1$, with $V=1$ and
$\gamma_t=\gamma_b=0$. Arrows in (a) specify legal column concatenation.
The quantum-input example in Sec.~\ref{subsec:conditional-inputs}
uses $N=8$}

\label{fig:exact-bridge}
\end{figure*}

\subsection{What survives for a quantum input}
\label{subsec:conditional-inputs}

A basis-path replacement does not yet specify the operation on a superposition
or mixed state, which can correlate the retained paths with the removed pieces.
The quantum operation is illustrated at $B=4$, $q=1$, $N=8$, $g^2=a=1$ and zero
source.

Let $b=(s_L,s_R)$ be the exposed rail pairs and $P_b$ the
corresponding path projector. A complete input path splits uniquely
into a retained path in $\mathcal K_b$ and removed legs in
$\mathcal L_b$. With their basis identification $U_b$,
\begin{equation}
\label{eq:general-input-space}
 \begin{aligned}
 & \mathcal H_{\rm in}^{\rm dom}\simeq
 \bigoplus_b(\mathcal K_b\otimes\mathcal L_b),\\
 & U_b:P_b\mathcal H_{\rm in}^{\rm dom}\to
       \mathcal K_b\otimes\mathcal L_b.
\end{aligned}
\end{equation}
Here the domain is the whole B4 open sector.
This coordinate decomposition imposes no product-state assumption.
The retained information is
\begin{equation}
\label{eq:conditional-block}
 \begin{aligned}
 & \sigma_b(\rho)=\operatorname{Tr}_{\mathcal L_b}
 (U_bP_b\rho P_bU_b^\dagger),\\
 & p_b=\operatorname{Tr}\sigma_b,\qquad \sum_b p_b=1.
\end{aligned}
\end{equation}
These matrices retain the original cut probabilities and the
coherence between compatible interior paths. They are not normalized
separately; correlations with the removed legs and coherence between
different $b$ are not retained.

Let $\mathcal J_b$ be the legal closing-path space and let $W_b$
concatenate a retained path and a closing path. Fix a normalized
basis bridge $|w_b\rangle$ through $c_1$ for every $b$, before
processing any input. The canonical reconstruction is
\begin{equation}
\label{eq:canonical-channel}
 \Phi(\rho)=\sum_b
 W_b[\sigma_b(\rho)\otimes|w_b\rangle\langle w_b|]W_b^\dagger.
\end{equation}
All branches contribute with their original weights.
The map is completely positive and trace preserving (CPTP), and
tracing the prepared bridge from each output block returns exactly
$\sigma_b$. Appendix~\ref{subsec:unified-kraus} proves these
properties for the full resource family.
Figure~\ref{fig:conditional-reconstruction}
shows the preserved conditional information and how the retained segment
joins the prepared bridge at fixed cuts.

\begin{figure*}[!tbp]
\centering
\includegraphics[width=0.5333333333\textwidth]{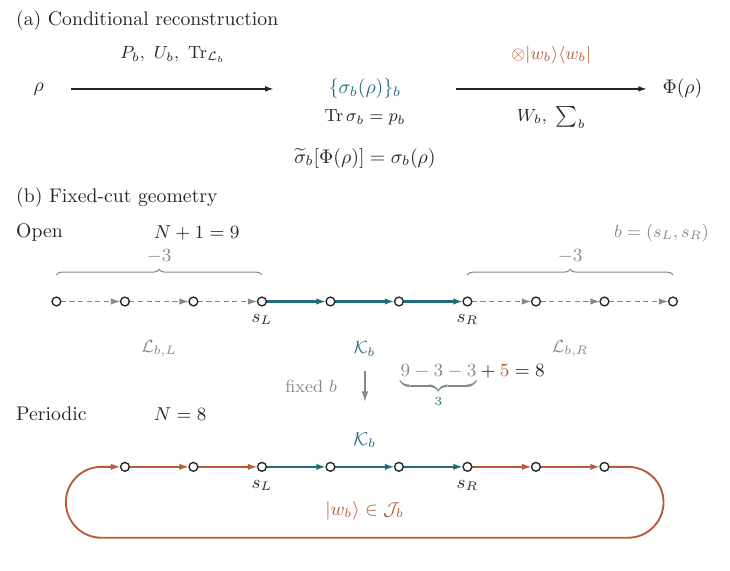}
\caption{Conditional reconstruction at $B=4$, $q=1$, $N=8$.
(a) The complete unnormalized interior matrices $\sigma_b$ retain
the original weights and within-branch coherence. A legal normalized
basis bridge is fixed before processing the input; summing all branches
gives the deterministic channel $\Phi$.
(b) For fixed cuts $b=(s_L,s_R)$, the gray legs $\mathcal L_{b,L/R}$
are removed. The retained segment $\mathcal K_b$ (teal) and the prepared
bridge in $\mathcal J_b$ (rust) form the periodic output at the same
physical length $Na$. Each directed edge, including the continuous
return, is one complete physical column}

\label{fig:conditional-reconstruction}
\end{figure*}

This block identity preserves physical interior measurements whose
active support changes neither the cuts nor the removed or prepared
pieces. Such an operator and its adjoint have the same retained
matrix in every compatible completion. Its expectation is therefore
the same pairing with $\sigma_b$ before and after reconstruction.
The support argument in Appendix~\ref{subsec:unified-identities}
extends this to the algebra generated by the matched operators,
without identifying it with all regional gauge-invariant observables.

\subsection{Coherence controls and the meaning of preservation}
\label{subsec:six-inputs}
\label{subsec:coherence-control}

Six fixed inputs distinguish where coherence is stored.
B is a basis path. S superposes paths with the same cut and removed
labels but different retained paths, so its interior cross term
survives. O instead differs in the removed labels, so tracing them
removes that cross term. X superposes different cut blocks.
M is an unequal mixture of S and X, and G is the numerical
zero-source open ground state. Further input conventions and numerical qualifications are provided in
Supplementary Information.

We compare complete-path dephasing $\mathcal D_{\rm path}$,
which deletes all off-diagonal path entries, with cut dephasing
$\mathcal D_\partial(\rho)=\sum_bP_b\rho P_b$. The three arms are
\begin{equation}
\label{eq:three-arms}
 \begin{aligned}
 \eta_A&=\Phi(\rho),\\
 \eta_B&=\Phi[\mathcal D_{\rm path}(\rho)],\\
 \eta_C&=\Phi[\mathcal D_\partial(\rho)].
\end{aligned}
\end{equation}
An arm label, the input code B and the cutoff $B$ are distinct.
For output extraction
$\widetilde\sigma_b(\eta)=
\operatorname{Tr}_{\mathcal J_b}(W_b^\dagger\eta W_b)$,
the conditional-block distance is
\begin{equation}
\label{eq:keep-distances}
 d_{\rm keep}(\rho,\eta)
 =\frac12\sum_b\|\widetilde\sigma_b(\eta)-\sigma_b(\rho)\|_1.
\end{equation}
Write $d_{\rm keep}^{\rm orig}(r)$ when the reference is the
original $\rho$ and $d_{\rm keep}^{\rm own}(r)$ when it is the
actual input $\rho_r$ to arm $r$. Extraction gives
$d_{\rm keep}^{\rm own}=0$ in every arm.
A and C also preserve the original blocks, since
$\Phi\mathcal D_\partial=\Phi$.
B preserves the blocks of its dephased input, which can differ
from the original ones. The distance tests the full conditional
matrices; it is not asserted to be optimal distinguishability in
the smaller physical observable algebra.

\begin{figure}[tbp]
\centering
\includegraphics[width=\linewidth]{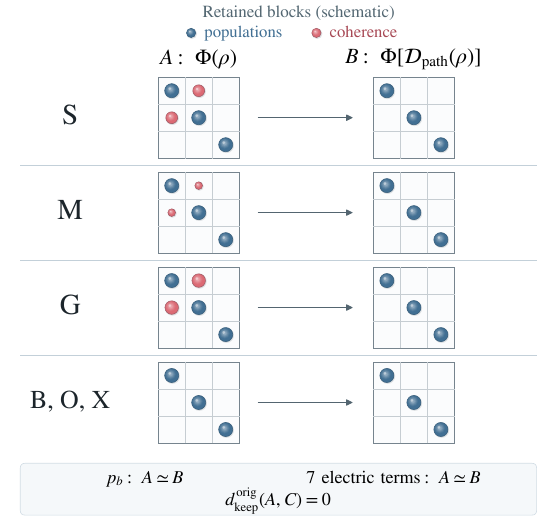}
\caption{Retained quantum information beyond path probabilities, for
the six inputs described in Sec.~\ref{subsec:six-inputs} at $B=4$,
$q=1$, $N=8$, $g^2=a=1$ and zero source. The four rows schematically
compare retained blocks in arms A and B of
Eq.~\eqref{eq:three-arms}; in arm B, complete-path dephasing acts
on the input before reconstruction. Blue markers indicate diagonal
entries and red markers within-block coherence. The grids and markers
do not encode matrix dimensions, amplitudes or phases.
B, O and X are distinct inputs grouped by their zero retained response;
O and X can have input coherence outside the retained blocks.
The corresponding distances and magnetic responses are given in
Table~\ref{tab:coherence-response}}

\label{fig:coherence-control}
\end{figure}

\begin{table*}[tb]
\centering
\caption{Original-input retained-block distances and signed magnetic
responses for the calculation in Fig.~\ref{fig:coherence-control}.
In $d_{\rm keep}^{\rm orig}(B)$, B denotes the processing arm.
For the matched interior magnetic terms,
$\Delta\langle M_i\rangle=\langle M_i\rangle_B-\langle M_i\rangle_A$.
Each of the distinct inputs B, O and X has the three zero responses
shown. Nonzero entries are rounded to four decimal places.}
\label{tab:coherence-response}
\begingroup
\renewcommand{\baselinestretch}{1}
\normalfont\fontsize{9}{11}\selectfont
\renewcommand{\arraystretch}{1}
\setlength{\tabcolsep}{0pt}
\setlength{\arrayrulewidth}{0.4pt}
\def\RTableRow{\rule[-8pt]{0pt}{22pt}}
\def\RTablePairRow{\rule[-13pt]{0pt}{32pt}}
\begin{tabular}{@{}c@{}c@{}c@{}c@{}}
\hline
\RTableRow\makebox[0.23\linewidth][c]{Input state} &
\makebox[0.23\linewidth][c]{$d_{\rm keep}^{\rm orig}(B)$} &
\makebox[0.23\linewidth][c]{$\Delta\langle M_0\rangle$} &
\makebox[0.23\linewidth][c]{$\Delta\langle M_1\rangle$} \\
\hline
\RTableRow S & $0.4583$ & $-0.3742$ & $0$ \\
\RTableRow M & $0.1237$ & $-0.1010$ & $0$ \\
\RTableRow G & $0.7648$ & $+0.6873$ & $+0.6873$ \\
\RTableRow B, O, X & $0$ & $0$ & $0$ \\
\hline
\end{tabular}
\endgroup
\end{table*}

For S, M and G, path dephasing preserves the branch probabilities
and diagonal electric expectations but changes interior magnetic
expectations. For B, O and X the retained response is zero;
coherence absent from a retained block cannot affect that block.
The opposite response signs for G and S/M show that dephasing
does not have a uniform direction of magnetic-energy change.
Classical path probabilities therefore do not determine this
quantum task.

Dephasing and reconstruction also have different energy references.
For arm $r$, preprocessing changes
$E_{\rm in}(\rho)$ to $E_{\rm in}(\rho_r)$, while the map cost is
$\Delta E_{\rm map}^{(r)}=E_{\rm out}(\eta_r)-E_{\rm in}(\rho_r)$.
Their sum is the total change from the original input.
The boundary budget derived below bounds the map cost.
The target excitation $G_{\rm var}=E_{\rm out}(\eta_r)-E_0^{\rm p}$
compares the output with the periodic ground energy.
The longitudinal comparison uses independently minimized open and
periodic energies.

\label{subsec:pure-target-separation}
Matched information leaves room for different global-state
properties. The canonical output for G has purity about $0.21337$,
while the outputs for B and S are pure. A physical boundary state can still be
chosen after the interior task is fixed. The energy change associated with this
choice is examined in Section~\ref{sec:boundary-resources}.

\section{Reconstruction and a uniform boundary energy cost}
\label{sec:general-reconstruction}

\subsection{Domains and exact lengths}
\label{subsec:general-domains}
\label{subsec:general-lengths}

At larger cutoffs, equal rail endpoints can support different vertical or
singlet witnesses, and the corresponding complete columns must remain distinct.
There is also a domain issue because legal periodic walks need not all be
reachable from the canonical rail pair.
Write
\begin{equation}
\label{eq:construction-domains}
\begin{aligned}
 \mathcal O_N&=\mathcal H_{N,B}^{(q,\mathrm o)},\\
 \mathcal P_N&=\operatorname{span}\left\{\begin{gathered}
 \text{labelled $N$-edge}\\
 \text{closed walks in the}\\
 \text{component of $c_q$}
 \end{gathered}\right\}.
\end{aligned}
\end{equation}
Every open canonical path lies in this component.
The reconstruction acts on all density operators in these domains
and their same-charge gluing products, including correlated
two-system inputs. The full periodic sector is larger at
$B=20/3$, where a noncanonical flux component is active.
Appendix~\ref{subsec:component-proof} proves domain invariance
and excludes that component from ground-state minimization, while
keeping its states in the full spectrum.
Hilbert-space fragmentation gives a related account of dynamically disconnected
sectors~\cite{CiavarellaBauerHalimeh2025Fragmentation}; the exclusion here
follows from the separate argument in that appendix.

The physical self-loop used in the minimal example also supplies
the general length adjustment. Complete directed routes lead from
each canonical-component node to $c_q$ and back. Repeating the loop
between those routes makes every sufficiently long prescribed
bridge. The witnesses allow a reserve $r=3$ at $B=4$ and $r=5$
at the other five cutoffs. These are sufficient choices, not
shortest-length claims; the proof is in
Appendix~\ref{subsec:exact-length-proof}.

\begin{table*}[tbp]
\caption{Four reconstructions at fixed $(B,q)$; $\mathcal O_N$ and
$\mathcal P_N$ are defined in Eq.~\eqref{eq:construction-domains}.
Brackets count complete graph edges. The sufficient thresholds are
$m,n\geq r$ for gluing and $N\geq2r-1$ for boundary changes, giving
physical lengths $(m+n)a$ and $Na$, respectively.}
\label{tab:four-operations}
\centering
\begingroup
\renewcommand{\baselinestretch}{1}
\normalfont\fontsize{9}{11}\selectfont
\renewcommand{\arraystretch}{1}
\setlength{\tabcolsep}{0pt}
\setlength{\arrayrulewidth}{0.4pt}
\def\RTableRow{\rule[-8pt]{0pt}{22pt}}
\def\RTablePairRow{\rule[-13pt]{0pt}{32pt}}
\begin{tabular}{@{}c@{}c@{}c@{}c@{}}
\hline
\RTableRow\makebox[0.27\linewidth][c]{Reconstruction} &
\makebox[0.29\linewidth][c]{Retained paths} &
\makebox[0.215\linewidth][c]{Inserted paths} &
\makebox[0.145\linewidth][c]{Total edges} \\
\hline
\RTablePairRow $\mathcal O_m\otimes\mathcal O_n\to\mathcal O_{m+n}$
 & \begin{tabular}[c]{@{}c@{}}$c_q\to x\ [m+1-r]$\\$y\to
c_q\ [n+1-r]$\end{tabular}
 & $x\to y\ [2r-1]$ & $m+n+1$\\
\RTablePairRow $\mathcal P_m\otimes\mathcal P_n\to\mathcal P_{m+n}$
 & \begin{tabular}[c]{@{}c@{}}$y_1\to x_1\ [m-r]$\\$y_2\to
x_2\ [n-r]$\end{tabular}
 & \begin{tabular}[c]{@{}c@{}}$x_1\to y_2\ [r]$\\$x_2\to y_1\ [r]$\end{tabular}
& $m+n$\\
\RTableRow $\mathcal O_N\to\mathcal P_N$
 & $x\to y\ [N+1-2r]$
 & $y\to x\ [2r-1]$ & $N$\\
\RTablePairRow $\mathcal P_N\to\mathcal O_N$
 & $x\to y\ [N-(2r-1)]$
 & \begin{tabular}[c]{@{}c@{}}$c_q\to x\ [r]$\\$y\to c_q\ [r]$\end{tabular} &
$N+1$\\
\hline
\end{tabular}
\endgroup
\end{table*}

Table~\ref{tab:four-operations} gives the four length accounts.
Open gluing removes the final $r$ columns of the first input and
the initial $r$ of the second. Periodic gluing removes $r$ columns
from each ring and uses cross bridges to form one ring.
A boundary change removes both open end pieces or a single
periodic piece, respectively. Cuts are fixed and conditioned on
jointly, so retained correlations survive. The common charge is
preserved, not added. Empty retained segments at a threshold have
coincident endpoints and one identity path; only realized joint
cut data are used.

\subsection{Fixed legal resources preserve the same information}
\label{subsec:general-channel}

The conditional decomposition of
Eq.~\eqref{eq:general-input-space} applies to each operation.
Now $b$ contains all exposed rails, $\mathcal L_b$ contains all
removed pieces and $\mathcal J_b$ contains the allowed replacement
tuples. Complete-path concatenation gives isometries $W_b$ with
orthogonal ranges for distinct $b$.
For fixed density operators $\tau_b$ on $\mathcal J_b$,
\begin{equation}
\label{eq:resource-channel}
\begin{aligned}
 &\mathcal R_\tau(\rho)
 =\sum_bW_b[\sigma_b(\rho)\otimes\tau_b]W_b^\dagger,\\
 &\operatorname{Tr}\tau_b=1.
\end{aligned}
\end{equation}
The choices include branches unoccupied by a particular input.
Appendix~\ref{app:channel-energy} proves CPTP, exact block extraction and matched physical
statistics for this whole family. The canonical channel is the
choice $\tau_b=|w_b\rangle\langle w_b|$.
No input factorization, branch selection or renormalization is used.
Operator-algebraic conditional expectations enter a different
constraint-quantization
framework~\cite{KlingerLeigh2024ConditionalExpectations}, which is not used to
construct this channel.

\subsection{Why the energy change stays near the boundary}
\label{subsec:energy-mechanism}

Reconstruction changes only a bounded neighborhood of the cuts.
More strongly, every unchanged interior Hamiltonian term has the
same conditional matrix before and after reconstruction.
Preserving $\sigma_b$ therefore cancels its contribution exactly.
This applies to all paired Hamiltonian terms, beyond the subset
of observables used in the controls.

Let $B_{\rm in}$ and $B_{\rm out}$ collect unpaired terms in the
removed or prepared pieces, their interaction-range neighborhoods,
affected boundary collars and the wrap where needed.
On the invariant input domain,
\begin{equation}
\label{eq:energy-defect}
 \mathcal R_\tau^\dagger(H_{\rm out})-H_{\rm in}
 =\mathcal R_\tau^\dagger(B_{\rm out})-B_{\rm in}.
\end{equation}
Here $\mathcal R_\tau^\dagger$ is the channel adjoint defined by
trace pairing. For gluing, $H_{\rm in}$ is the sum Hamiltonian.
The output is tested with the full physical Hamiltonian, which
need not preserve the matching range.
Positivity and unitality of the adjoint then give
\begin{equation}
\label{eq:local-energy-budget}
 |\Delta E_{\rm map}|
 \leq\|\mathcal R_\tau^\dagger(H_{\rm out})-H_{\rm in}\|
 \leq K_{B,q,\mathrm{op}} .
\end{equation}
Only finitely many bounded local terms remain.
Adding retained length adds matched terms, not new boundary
neighborhoods, so $K_{B,q,\mathrm{op}}$ is uniform in the allowed
lengths and fixed legal resource choices.
Appendix~\ref{subsec:energy-proof} supplies the norm bound.
Figure~\ref{fig:energy-cancellation} shows this mechanism.

\begin{figure*}[tbp]
\centering
\includegraphics[width=1.0\textwidth]{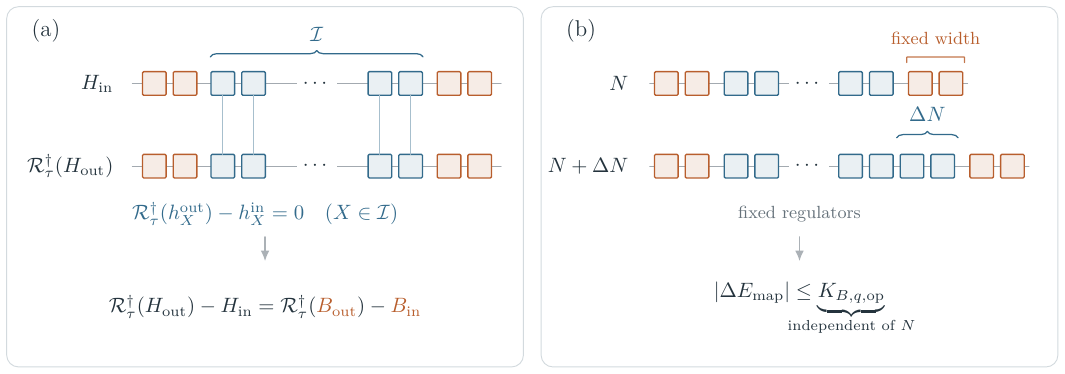}
\caption{Interior cancellation and a length-independent boundary
energy budget. (a) Blue squares denote paired interior Hamiltonian
terms indexed by $\mathcal I$. In the lower row, the output terms
are pulled back by $\mathcal R_\tau^\dagger$. Their differences vanish
on the invariant input domain. Orange portions denote the unpaired
sums $B_{\rm in}$ and $\mathcal R_\tau^\dagger(B_{\rm out})$, including
the removed or prepared pieces and their affected neighborhoods
[Eq.~\eqref{eq:energy-defect}].
(b) Increasing the retained length adds matched terms while the
affected neighborhoods retain uniformly bounded widths.
At fixed regulator parameters and operation, the upper bound
$K_{B,q,\mathrm{op}}$ in Eq.~\eqref{eq:local-energy-budget} is uniform
in the allowed lengths. Counts and distances are schematic; the actual
energy change may depend on length and input state}

\label{fig:energy-cancellation}
\end{figure*}

The budget is at fixed regulator and does not assert a small or
optimized actual cost. The budget can be used both to lower boundary energy with
the interior information fixed and to compare variational energies at different
lengths or boundaries.

\section{Boundary energy at fixed interior information}
\label{sec:boundary-resources}

The prepared boundary state can still be chosen under exact block preservation.
The output energy can change with this choice because the full Hamiltonian
contains boundary and interface terms, while the matched interior expectations
remain fixed.
For a reference input $\rho_\star$, set $S_b=\sigma_b(\rho_\star)$.
The resource objective is
\begin{equation}
\label{eq:boundary-resource-objective}
\begin{aligned}
 & G_b=\operatorname{Tr}_{\mathcal K_b}
 [(S_b\otimes\mathbf1)W_b^\dagger H_{\rm out}W_b],\\
 & E_{\rm out}(\rho_\star;\tau)=\sum_b\operatorname{Tr}(G_b\tau_b).
\end{aligned}
\end{equation}
The original branch weight is included in $S_b$.
The Hermitian objective attains its minimum on a pure vector in
its lowest eigenspace, as shown in
Appendix~\ref{subsec:resource-objective-proof}.
In general that direction depends on $\rho_\star$.
Once chosen, it must be held fixed to define one linear channel
for subsequent inputs. The optimization is within this legal
product-preparation family, not over all physical outputs.

\subsection{Why the B4 choice works for different inputs}
\label{subsec:b4-resource-designs}

The $B=4$, $q=1$, $N=8$, $g^2=a=1$ example is considered again.
At fixed cuts, the retained path has two free rail nodes and the
five-column bridge has four. Since each can be $x,y,c_1$, their
spaces have dimensions nine and 81.
A B4 plaquette cycles the middle rail node of its two-column
neighborhood. The two plaquettes across the cuts change the cut
labels and therefore have zero same-branch compression.
The remaining magnetic terms act inside either the retained path
or the bridge. Including every electric and interface term gives
\begin{equation}
\label{eq:b4-resource-factorization}
 W_b^\dagger H^{\rm p}W_b
 =K_b^{\rm int}\otimes I_{81}+I_9\otimes T_b .
\end{equation}
This support and selection-rule argument is completed in
Appendix~\ref{subsec:b4-factorization-proof}.
It is specific to this B4 geometry; no such separation is assumed
at a larger cutoff. The cross-cut terms remain in the full
Hamiltonian, with zero expectation in the branch-diagonal output.

Consequently
\begin{equation}
\label{eq:b4-reference-objective}
 G_b=\operatorname{Tr}(S_bK_b^{\rm int})I_{81}+p_bT_b,
 \qquad p_b=\operatorname{Tr}S_b .
\end{equation}
For $p_b>0$, the scalar shift and positive multiplier leave the
preferred direction of $T_b$ unchanged. The boundary choice can
therefore be fixed independently of the evaluated input in this
example, including a normalized choice for every zero-weight branch.

We compare the canonical basis bridge $|w_b\rangle$, the
lowest-energy basis path, and a coherent pure bridge selected from
the lowest numerical eigenspace of $T_b$. The same fixed choices
are used for all six inputs; their numerical prescriptions and
qualification are given in Supplementary Information.
For a fixed pure family $|v_b\rangle$, subtraction cancels the
internal contribution:
\begin{equation}
\label{eq:resource-gain-weighted}
 \begin{aligned}
 &E_{\rm can}(\rho)-E_v(\rho)\\
 &\quad=\sum_b p_b(\rho)
 [\langle w_b|T_b|w_b\rangle-\langle v_b|T_b|v_b\rangle].
\end{aligned}
\end{equation}
The input supplies only its original weights.
Figure~\ref{fig:resource-gains} shows the local gains and their
weighted sums. Some basis-path choices give no gain, whereas
coherent bridges use magnetic couplings between paths to lower
every branch target in this example. Equal branch weights also
explain the equal total gains for B, S and O.

\begin{figure*}[tbp]
\centering
\includegraphics[width=0.96\textwidth]{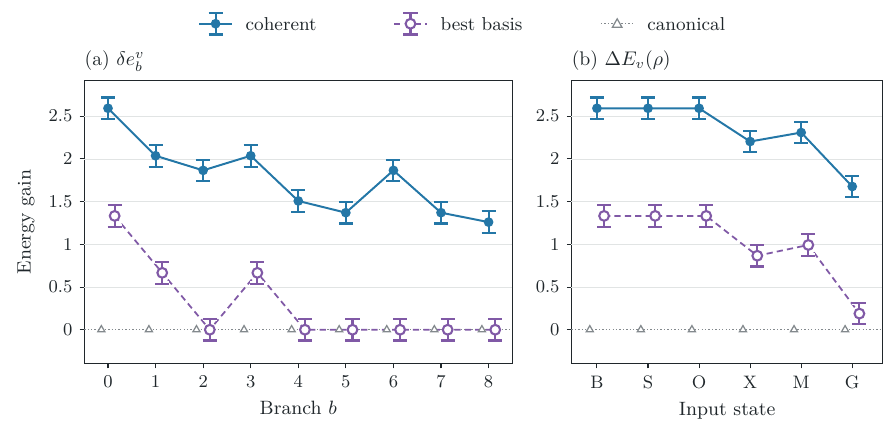}
\caption{Boundary energy gains at $B=4$, $q=1$, $N=8$ and $g^2=a=1$.
(a) Local gains
$\delta e_b^v=\langle w_b|T_b|w_b\rangle-\langle v_b|T_b|v_b\rangle$
for branches $b=3s_L+s_R$, using node indices $0,1,2$ for
$x,y,c_1$ of Eq.~\eqref{eq:b4-vertices}.
(b) Gains $\Delta E_v(\rho)=E_{\rm can}(\rho)-E_v(\rho)$ relative to
the canonical output, obtained from the same fixed preparations and
the original branch weights through
Eq.~\eqref{eq:resource-gain-weighted}.
Both panels use the same energy scale.
Blue filled circles and solid lines denote coherent preparations;
purple open circles and dashed lines denote best-basis preparations;
gray triangles mark the canonical zero reference.
Circle centers give the values. The fixed-size stems and caps are
marker elements, not uncertainty intervals.
Horizontal offsets separate categorical markers, and connecting
lines guide the eye}

\label{fig:resource-gains}
\end{figure*}

For G, the coherent choice lowers the canonical output energy
from about $5.598$ to $3.920$, still above the independent
periodic ground energy of about $2.588$. The gain concerns
the fixed input and the chosen legal resource family.

\subsection{Energy improvement without purification}
\label{subsec:resource-mixedness}

For pure resources, changing the boundary state leaves the nonzero
output spectrum unchanged. Orthogonal branches give
\begin{equation}
\label{eq:resource-purity}
 \operatorname{Tr}[\mathcal R_\tau(\rho)^2]
 =\sum_b\operatorname{Tr}(\sigma_b^2)\operatorname{Tr}(\tau_b^2).
\end{equation}
Appendix~\ref{subsec:resource-spectrum-proof} gives the spectral
argument. All three preparations use pure resources and have the
same output purity at fixed input. Boundary coherence lowers G's
energy while its output remains mixed.

\section{Center-flux energetics and the longitudinal limit}
\label{sec:longitudinal}

The uniform boundary cost allows the reconstruction to be used for a variational
comparison. A state of one ladder is used to construct a legal trial state of
another, with an energy overhead that does not grow with retained length.
Any fixed legal resource suffices; energy optimization is not required.
Lattice ground-state energy densities have also been studied by other
Hamiltonian
methods~\cite{HollenbergWitte1994EnergyDensity,%
BronzanGunal1994PureSU3LargeLattice}.

\subsection{Gluing establishes the sector densities}
\label{subsec:sector-energies}
\label{subsec:sector-limits}

Fix one of the six cutoffs, $a>0$, $g^2>0$ and the transverse width.
The local and domain conditions H1--H4 are given in
Appendix~\ref{subsec:hypotheses}.
Let $E_{0,B}^{(s,\alpha)}(Na)$ be the full-sector ground energy
for $s=\mathrm{vac},\mathrm{flux}$ ($q=0,1$) and
$\alpha=\mathrm o,\mathrm p$. The vacuum-subtracted energy is
\begin{equation}
\label{eq:flux-excess}
 \Delta E_B^{(\alpha)}(Na)
 =E_{0,B}^{(\mathrm{flux},\alpha)}(Na)
  -E_{0,B}^{(\mathrm{vac},\alpha)}(Na).
\end{equation}
Every comparison uses the same regulator, length and Hamiltonian
normalization. The periodic minimum can be approached in the
canonical operation domain: component invariance and the strict
energy exclusion of the exceptional $B=20/3$ flux block establish
this in Appendix~\ref{subsec:exception-energy}.
Thus the following statements concern full-sector minima.

Glue two approximate minimizing states in the same sector and
boundary realization. Their product is a sufficient trial input;
the channel itself also permits correlations.
The first two operations in Table~\ref{tab:four-operations}
and the local energy budget give
\begin{equation}
\label{eq:sector-gluing}
 \begin{aligned}
 E_{0,B}^{(s,\alpha)}((m+n)a)
 &\leq E_{0,B}^{(s,\alpha)}(ma)\\
 &\quad+E_{0,B}^{(s,\alpha)}(na)+C_B^{(s,\alpha)}
\end{aligned}
\end{equation}
for both $m,n$ in the gluing tail, with a constant independent
of their lengths. The bounded local Hamiltonian also gives
$E_0(Na)\geq-AN-D$.
After adding $C_B^{(s,\alpha)}$, the energy is subadditive on
that tail. Appendix~\ref{subsec:stability-tail} proves the
limit while keeping the final remainder block within the allowed
tail.

\begin{theorem}[Existence of fixed-cutoff sector and line-energy
densities]\label{thm:bulkexist}
Under H1--H4, for every cutoff studied here and $\alpha\in\{\mathrm o,\mathrm
p\}$, the sector densities
\[
e_B^{(s,\alpha)}
=
\lim_{N\to\infty}\frac{E_{0,B}^{(s,\alpha)}(Na)}{Na},
\qquad s\in\{\mathrm{vac},\mathrm{flux}\},
\]
exist and are finite.  Consequently,
\[
\sigma_B^{(\alpha)}
=
e_B^{(\mathrm{flux},\alpha)}-e_B^{(\mathrm{vac},\alpha)}
=
\lim_{N\to\infty}\frac{\Delta E_B^{(\alpha)}(Na)}{Na}
\]
exists and is finite.
\end{theorem}

The argument is applied to each sector separately. It does not
assume subadditivity of the vacuum-subtracted difference.
The proof first gives energy per plaquette; dividing by fixed
$a$ gives energy per physical length.
Infinite-lattice QCD~\cite{AddIV05,AddIV06} and charge or boundary
algebras~\cite{AddII08,AddII12,AddII23,AddII19,AddVII02} provide related
settings, distinct from this finite-ladder construction.
Local stability in weakly interacting, locally gapped spin
systems~\cite{HenheikTeufelWessel2022LocalStability} is another comparison;
the present bound uses the explicit reconstruction budget.

\subsection{Opening and closing change only a boundary energy}
\label{subsec:boundary-independence}

The last two operations in Table~\ref{tab:four-operations}
compare open and periodic ladders at the same $Na$.
Their sufficient threshold is $N_0=2r-1$, namely five at $B=4$
and nine at the other cutoffs.
Apply each map to an approximate minimizer in its own domain.
The two variational comparisons bound the periodic energy by
the open one plus a constant, and the open energy by the periodic
one plus another constant. These maps need not be inverses.
Appendix~\ref{subsec:boundary-proof} proves both inequalities
and uses the already existing limits to obtain the second result.

\begin{theorem}[Fixed-cutoff open/periodic boundary-condition
independence]\label{thm:boundary}
Under H1--H4, for every cutoff studied here and
$s\in\{\mathrm{vac},\mathrm{flux}\}$, there is a finite constant
$C_{\partial,B}^{(s)}$, independent of $N$, such that
\[
\left|E_{0,B}^{(s,\mathrm p)}(Na)-E_{0,B}^{(s,\mathrm o)}(Na)\right|
\le C_{\partial,B}^{(s)}
\]
for all sufficiently large $N$.  Therefore
\[
e_B^{(s,\mathrm p)}=e_B^{(s,\mathrm o)},
\qquad
\sigma_B^{\mathrm p}=\sigma_B^{\mathrm o}.
\]
\end{theorem}

At finite matched length the triangle inequality also gives
\begin{equation}
\label{eq:boundary-density-budget}
\begin{aligned}
 &\left|\frac{\Delta E_B^{(\mathrm p)}(Na)
                 -\Delta E_B^{(\mathrm o)}(Na)}{Na}\right|\\
 &\quad\leq\frac{C_{\partial,B}^{(\mathrm{flux})}
                  +C_{\partial,B}^{(\mathrm{vac})}}{Na}.
\end{aligned}
\end{equation}
This bounds disagreement between boundaries, not the distance
of either density from its own limit.
Studies of twisted partition functions provide other measures of boundary
sensitivity~\cite{Kanazawa2009TomboulisYaffeSUN,AddII01,AddII07}.

\begin{figure*}[!tbp]
\centering
\includegraphics[width=0.6266666667\textwidth]{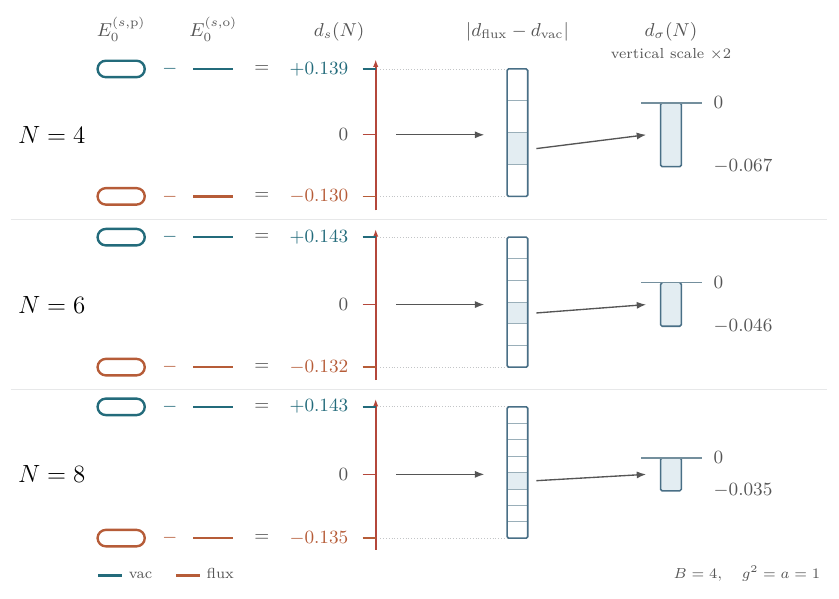}
\caption{Boundary differences from independent full-sector ground
energies at $B=4$ and $g^2=a=1$. Rows show $N=4,6,8$.
Teal denotes vacuum and rust fundamental flux; capsules and strips
denote periodic and open boundaries. Colored differences are positioned
by $d_s$ in Eq.~\eqref{eq:same-length-differences}, not by individual
ground energies. Each central band represents $|d_{\rm flux}-d_{\rm vac}|$,
divided into $N$ equal cells of $a|d_\sigma|$. Each right bar extends down
from its top-edge zero to negative $d_\sigma$, with twice the cell height
at $a=1$. Cells and arrows show arithmetic, not spatial partitions or
state transformations. The $N=4$ row is below the sufficient channel
threshold $N_0=5$}

\label{fig:matched-spectra}
\end{figure*}

\subsection{Finite spectra and physical scope}
\label{subsec:same-length-diagnostic}
\label{subsec:longitudinal-scope}

For $B=4$ and $g^2=a=1$, independently minimized full-sector
energies at $N=4,6,8$ give
\begin{equation}
\label{eq:same-length-differences}
\begin{aligned}
 d_s(N)&=E_{0,4}^{(s,\mathrm p)}(Na)-E_{0,4}^{(s,\mathrm o)}(Na),\\
 d_\sigma(N)&=[d_{\rm flux}(N)-d_{\rm vac}(N)]/(Na).
\end{aligned}
\end{equation}
Figure~\ref{fig:matched-spectra} shows positive vacuum shifts
and negative flux shifts. The density differences are negative
and decrease in magnitude across these three samples. These finite
cases illustrate boundary dependence; they do not determine a
convergence rate.
Minimally truncated and finite-lattice SU(3) spectra provide related
benchmarks~\cite{ChenMullerYao2026,ChinLongRobson1988SU3GroundState,%
HamerSamarasBursill2000SU3GFMC}.

The results establish a center-flux line-energy density at fixed
regulator and cutoff. Its positivity, cutoff removal, increasing
transverse extent and continuum interpretation require further
control.
For comparison, scaling has been studied with finite representations and
improved
actions~\cite{AddIV03,DeGrandEtAl1996FixedPointSU3,%
ChinLongRobson1986GlueballScaling}, and with digitized SU(3)
theories~\cite{AlexandruBedaqueBrettLamm2022DigitizedQCD} or qubit-regularized
and partitioned SU(2) theories~\cite{SiewChandrasekharanBhattacharya2026SU2Qubit,%
JakobsEtAl2025PartitioningsSU2}.
A static-source interpretation additionally
requires
the same central bulk density and endpoint or contact terms
subextensive in source separation; fixing open rail endpoints
does not itself add static quarks.
Static-source flux-tube spectra~\cite{AddVI24,AddVI25} give a distinct physical
comparison.

\section{Conclusions}
\label{sec:conclusions}

Boundary reconstruction in the regulated SU(3) ladder can preserve
the quantum information needed for an interior physical task.
The exposed representations are matched using legal column paths, while the
required length is supplied by physical self-loops. The original branch weights
and within-branch coherence are preserved by conditioning on the cuts. The
dephasing controls show why retaining only path
probabilities would lose magnetic information.

The boundary state remains a separate choice. In the B4 example, the output
energy is lowered by coherent legal resources while the conditional blocks and
pure-resource output spectrum remain unchanged.
This energy improvement is distinct from purification and global ground-state
preparation.

For the longitudinal problem, the essential consequence is exact
interior cancellation. Same-sector gluing and bidirectional comparisons at equal
length are obtained from a uniform boundary budget. Same-sector gluing
establishes the sector densities, while the bidirectional comparisons make their
vacuum-subtracted difference independent of open or periodic boundaries. Both
statements are at fixed regulator,
with the physical scope stated in
Sec.~\ref{subsec:longitudinal-scope}.

\begin{appendices}
\counterwithin*{equation}{section}
\displayappendixsection{Conditional reconstruction and its energy cost}
\label{app:channel-energy}

Fix an operation in Table~\ref{tab:four-operations}, its cuts and
its allowed length. The domains are those of
Appendix~\ref{app:construction}. We use the finite-dimensional
operator-sum characterization of
channels~\cite{Kraus1971GeneralStateChanges,Choi1975,Watrous2018}.

\subsection{Physical coordinates and one channel proof}
\label{subsec:conditional-proof}
\label{subsec:unified-kraus}

At fixed realized cut data $b$, cutting a complete path gives a unique
retained tuple $k$ and removed tuple $\lambda$. Conversely, every
compatible pair concatenates: the exposed rails are fixed and each
column already contains both singlet witnesses. This bijection of
orthonormal bases gives $U_b$ in
Eq.~\eqref{eq:general-input-space}, including all multiplicities.
Concatenation with a legal new tuple similarly gives $W_b$.
Fixed segment positions recover $k$, the new tuple and $b$ from the
output, so $W_b$ is an isometry and different $b$ have orthogonal
ranges. For joint gluing, these coordinates include both retained
paths and both removed pieces; tracing the latter preserves
correlations between the retained paths.

Write each fixed normalized resource as
$\tau_b=\sum_\mu t_{b\mu}|j_{b\mu}\rangle\langle j_{b\mu}|$,
where $t_{b\mu}\geq0$ and $\sum_\mu t_{b\mu}=1$.
For an orthonormal basis of $\mathcal L_b$, define
\begin{equation}
\label{eq:unified-resource-kraus}
 L_{b\lambda\mu}=\sqrt{t_{b\mu}}\,
 W_b(\mathbf1\otimes|j_{b\mu}\rangle)
 (\mathbf1\otimes\langle\lambda|)U_bP_b.
\end{equation}
Their operator sum is Eq.~\eqref{eq:resource-channel}.
It is completely positive, has legal same-charge outputs, and obeys
\begin{equation}
\label{eq:unified-completeness}
 \begin{aligned}
 &\sum_{b,\lambda,\mu}L_{b\lambda\mu}^\dagger L_{b\lambda\mu}\\
 &\quad=\sum_bP_bU_b^\dagger
   (\mathbf1_{\mathcal K_b}\otimes\mathbf1_{\mathcal L_b})U_bP_b\\
 &\quad=\mathbf1_{\mathcal H_{\rm in}^{\rm dom}}.
\end{aligned}
\end{equation}
Thus it preserves trace on the whole declared domain, including
inputs correlated across gluing pieces. Resources are specified even
on branches with zero weight for the input being evaluated.
Orthogonality and $\operatorname{Tr}\tau_b=1$ give
\begin{equation}
\label{eq:general-extraction}
 \begin{aligned}
 & W_b^\dagger\mathcal R_\tau(\rho)W_b=\sigma_b\otimes\tau_b,\\
 & \operatorname{Tr}_{\mathcal J_b}
 [W_b^\dagger\mathcal R_\tau(\rho)W_b]=\sigma_b.
\end{aligned}
\end{equation}
There is no division by a branch weight. The tensor product is created
by preparation, rather than assumed for the input. Taking
$\tau_b=|w_b\rangle\langle w_b|$ proves the canonical case as well.

\subsection{Matched observables and interior cancellation}
\label{subsec:unified-identities}
\label{subsec:energy-proof}

Choose physical generators whose active column and intertwiner support
stays in the retained region. They may use exposed rails as fixed
controls, but change neither those labels nor the old or new pieces.
For every compatible completion, their retained action is the same
matrix $\mathsf a_{\mu,b}$. The support statement also holds for each
adjoint: all changed retained paths remain in $\mathcal K_b$ and no
changed completion is produced. Hence
\begin{equation}
\label{eq:general-generator-intertwining}
 \begin{aligned}
 & O_\mu^{\rm out}W_b
 =W_b(\mathsf a_{\mu,b}\otimes\mathbf1_{\mathcal J_b}),\\
 & U_b O_\mu^{\rm in}P_b
 =(\mathsf a_{\mu,b}\otimes\mathbf1_{\mathcal L_b})U_bP_b.
\end{aligned}
\end{equation}
The output range is therefore reducing for these generators.
Products and adjoints give a matched physical algebra, which need
not be all regional gauge-invariant observables or all matrices on
$\bigoplus_b\mathcal K_b$. For every matched pair in that algebra,
\begin{equation}
\label{eq:unified-trace-pairing}
 \operatorname{Tr}[O^{\rm out}\mathcal R_\tau(\rho)]
 =\sum_b\operatorname{Tr}(\mathsf a_b\sigma_b)
 =\operatorname{Tr}(O^{\rm in}\rho).
\end{equation}

Apply the same support argument to every paired interior Hamiltonian
term, whether or not it belongs to the smaller set of measured
observables. Trace pairing for all domain density operators gives
$\mathcal R_\tau^\dagger(h_X^{\rm out})=h_X^{\rm in}$ on that domain.
Summing cancels the interior and proves
Eq.~\eqref{eq:energy-defect}. The output Hamiltonian is the full
physical operator; it need not preserve the matching range.

Trace preservation makes the adjoint unital, and complete positivity
makes it positive. For self-adjoint $h$,
\begin{equation}
\label{eq:positive-order-bound}
\begin{aligned}
 &-\|h\|\mathbf1\leq h\leq\|h\|\mathbf1\\
 &\quad\Longrightarrow\quad
 \|\mathcal R_\tau^\dagger(h)\|\leq\|h\|.
\end{aligned}
\end{equation}
Magnetic terms are grouped as
$-g^{-2}(\Box+\Box^\dagger)$ before this estimate is used.
The triangle inequality bounds the two remaining boundary sums.
At fixed regulator, H1--H3 give a common local norm bound
$h_{\max}$ and a finite interaction range. A fixed number of cuts
and the bounded replacement lengths in
Table~\ref{tab:four-operations} leave at most
$n_{\partial,B,q,\mathrm{op}}$ unpaired terms on both sides together.
Thus a sufficient choice in Eq.~\eqref{eq:local-energy-budget} is
\begin{equation}
\label{eq:uniform-budget-choice}
 K_{B,q,\mathrm{op}}
 =n_{\partial,B,q,\mathrm{op}}h_{\max}.
\end{equation}
It is independent of the input, longitudinal lengths and fixed legal
resource direction. Overlapping short boundary neighborhoods cannot
increase this count. The argument also applies to the sum Hamiltonian
of two inputs. The constant may depend on the fixed cutoff and
couplings; it bounds map cost, not excitation above a target ground
state.

\subsection{The resource objective and the B4 separation}
\label{subsec:resource-objective-proof}

Set $H_{bb}=W_b^\dagger H_{\rm out}W_b$ and
$S_b=\sigma_b(\rho_\star)$. Branch diagonality of the output gives
$E_{\rm out}=\sum_b\operatorname{Tr}[H_{bb}(S_b\otimes\tau_b)]$,
which is Eq.~\eqref{eq:boundary-resource-objective}.
This uses the full Hamiltonian, including interface couplings.
For every resource vector $v$,
$v^\dagger G_bv=\operatorname{Tr}[H_{bb}
(S_b\otimes|v\rangle\langle v|)]$ is real, so $G_b$ is Hermitian.
Its expectation in a density operator is a convex combination of
eigenvalues; a minimum is attained by a pure vector in its lowest
eigenspace. In general this direction depends on the reference input.

\label{subsec:b4-factorization-proof}
For the B4 example, label consecutive periodic rail nodes by
$(s_L,u_1,u_2,s_R,v_1,v_2,v_3,v_4)$, returning to $s_L$.
The $u$ nodes are retained and the $v$ nodes belong to the new bridge.
Each node is one of $x,y,c_1$ in
Eq.~\eqref{eq:b4-vertices}; there is a unique complete edge
between any adjacent pair. A forward plaquette changes the middle
node of its two-column neighborhood around $x\to y\to c_1\to x$,
and its adjoint reverses the cycle. It cannot leave that node fixed.
The two plaquettes at $s_L,s_R$ therefore connect different cut
blocks and have zero same-$b$ compression. The remaining magnetic
terms act either on the $u$ nodes or on the $v$ nodes.
All electric terms are diagonal: links in the retained path depend
only on $u_1,u_2$ and fixed cuts, while links in the bridge depend
only on the $v$ nodes and fixed cuts. Assign the latter, including
the fixed-interface rail terms, to $T_b$. This classifies every
term of the full Hamiltonian and gives
Eq.~\eqref{eq:b4-resource-factorization}, without omitting the
cross-cut magnetic operators or an interface electric term.

Substitution yields Eq.~\eqref{eq:b4-reference-objective}.
For $p_b>0$, the positive multiplier and scalar shift do not change
the minimizing direction of $T_b$. If $p_b=0$, positivity implies
$S_b=0$; that branch contributes nothing for this reference but still
has its fixed normalized resource. Subtracting two resource
expectations cancels the internal scalar and proves
Eq.~\eqref{eq:resource-gain-weighted}. Exact branch minima give
nonnegative gains within this family. The numerical rays used in
the figures are approximations, not rigorous eigenvalue enclosures.

\subsection{Energy changes at fixed output spectrum}
\label{subsec:resource-spectrum-proof}

Let $s_{b\alpha}$ and $t_{b\mu}$ be the eigenvalues of $\sigma_b$
and $\tau_b$. Isometry and orthogonal branch ranges give
\begin{equation}
\label{eq:resource-output-spectrum}
 \operatorname{spec}_{\ne0}\mathcal R_\tau(\rho)
 =\bigsqcup_b
   \{s_{b\alpha}t_{b\mu}:s_{b\alpha}t_{b\mu}>0\},
\end{equation}
with multiplicities. Squaring and summing proves
Eq.~\eqref{eq:resource-purity}. For pure resources, the nonzero
spectrum is the union of the spectra of the original $\sigma_b$,
independent of the resource directions. Mixed resources cannot
increase the purity because $\operatorname{Tr}\tau_b^2\leq1$.

\displayappendixsection{Physical domains, exact lengths and the two limits}
\label{app:construction}

\subsection{Fixed-regulator conditions}
\label{subsec:hypotheses}

\phantomsection\label{hyp:H1}\noindent\textbf{(H1)} At fixed regulator, coupling
convention, finite transverse volume, and cutoff, the self-adjoint finite-volume
Hamiltonians are the size, sector, and boundary restrictions of the common
$N$-independent local interaction rule specified above; realization-dependent
terms lie in a fixed-width boundary collar and are uniformly bounded in $N$.

\phantomsection\label{hyp:H2}\noindent\textbf{(H2)} The Gauss-constrained
sectors are Hamiltonian invariant.

\phantomsection\label{hyp:H3}\noindent\textbf{(H3)} At fixed $B$, the local
state space and intertwiner multiplicities are finite, while the Hamiltonian has
finite interaction range and uniformly bounded local terms.

\phantomsection\label{hyp:H4}\noindent\textbf{(H4)} In the strongly connected
component of each $\mathcal G_{B,q}$ that contains $c_q$, every vertex has
directed paths to and from $c_q$, which has a physical self-loop; the only
noncanonical active flux components in the six-cutoff census occur at $B=20/3$,
one in each of the $q=1,2$ sectors, and are treated separately in
Appendix~\ref{appsub:component-domain}.

The complete columns retain both independent intertwiner labels.
Finite directed-graph witnesses for all eighteen $(B,q)$ graphs and the
exceptional-block certificate were used in verifying these finite statements. They supply finite local witnesses, while the arguments below
extend them to the required lengths and energy comparisons.

\subsection{Invariant domains and the exceptional flux block}
\label{subsec:component-proof}
\label{appsub:component-domain}

An open $c_q\to c_q$ path lies in the strongly connected component
containing $c_q$: every visited node is reached along its prefix
and returns along its suffix. A periodic walk likewise lies in one
strongly connected component. The electric Hamiltonian is diagonal.
A local magnetic update fixes its two outer graph nodes, and the
unchanged remainder of the periodic walk returns from the right
node to the left. A new intermediate node is therefore mutually
reachable with these nodes. The updated walk stays in the same
component. For an open walk, its canonical endpoints give the
corresponding return connection. Together with H2, this proves
invariance of the declared domains and the periodic decomposition
into component blocks. Graph edges describe concatenation, not
Hamiltonian matrix elements. Zero-degree candidate vertices support
no paths.

\label{subsec:exception-energy}
The only noncanonical active flux block in this cutoff family occurs
at $B=20/3$. In $q=1$ it repeats the node $(\bar6,6)$ with the unique
self-loop $(V,\gamma_t,\gamma_b)=(1,0,0)$.
Nonconstant magnetic updates would require a local
$(6,8,3)$ or $(8,6,\bar3)$ singlet, or a conjugate, whose Casimir
sum is $23/3>20/3$ and is excluded. Independently, the retained
forward plaquette kernel has no nonzero entry on the uniform two-column
pattern, as specified in the exceptional-block certificate.
Its adjoint vanishes there as well. This zero action is not inferred
from the block being one dimensional.

Each periodic cell then has electric Casimir sum $20/3$.
The straight fundamental trial state in the canonical component
has sum $4/3$ per cell and zero magnetic expectation. In the common
Hamiltonian convention,
\begin{equation}
\label{eq:exception-energy-comparison}
 \begin{aligned}
 & E_{\rm iso}(N)=\frac{10}{3}g^2N,\\
 & E_{0,\rm main}^{\rm flux}(N)\leq\frac23g^2N,\\
 & E_{\rm iso}-E_{0,\rm main}^{\rm flux}\geq\frac83g^2N>0.
\end{aligned}
\end{equation}
Charge conjugation gives the $q=2$ result. The exceptional states
remain in the full spectrum, but cannot minimize it for $g^2>0$.
There is no additional active vacuum component in the stated family.
Consequently the full-sector infima needed below can be approached
within the canonical domains; the maps are not extended to arbitrary
noncanonical inputs.

\subsection{Padding and legality of the four operations}
\label{subsec:exact-length-proof}
\label{subsec:four-operation-proof}

Complete directed witnesses connect every canonical-component node
to and from $c_q$, with maximal distance $d_{B,q}=1$ at $B=4$ and
at most two at the other five cutoffs. Every step carries its
vertical irrep and both singlet labels. The physical loop at $c_q$
has $V=1$ and $\gamma_t=\gamma_b=0$.
For routes $x\to c_q$ and $c_q\to y$ of lengths
$\ell_x,\ell_y$, insert $\ell-\ell_x-\ell_y$ copies of that loop.
For every $\ell\geq\ell_x+\ell_y$ this is a legal path of exactly
$\ell$ columns. In particular, the reserves $r=3,5$ satisfy
$2d_{B,q}<r$, so every bridge of length $r$ or $2r-1$ in
Table~\ref{tab:four-operations} exists. Legs ending at $c_q$
are padded there in the same way.

Table~\ref{tab:four-operations} gives the length accounts.
For periodic gluing, let retained paths run from $y_i$ to $x_i$.
The two bridges go from $x_1$ to $y_2$ and from $x_2$ to $y_1$.
Their cross concatenation makes one ring with the original graph
charge $q$, not two rings or the sum of two charges. Cuts are fixed
before the input is processed and conditioned on jointly.
The decomposition includes only jointly realized cut labels.
If a retained segment is empty at a threshold, its endpoints
coincide and its path space is the single identity witness at that
node. It is not an edge-free connection between distinct nodes.
These conventions preserve the path bijection used in Appendix~\ref{app:channel-energy}
also at equality in the sufficient thresholds.

\subsection{Gluing, stability and existence of densities}
\label{subsec:gluing-variational}
\label{subsec:stability-tail}

Fix $B,s,\alpha$ and all regulator parameters. Choose fixed legal
resources, and take $m,n$ in the gluing tail.
By the domain argument above, normalized states $\rho_m,\rho_n$
in the operation domains can approximate their full-sector
infima within any $\varepsilon>0$. Their product is a variational
choice, not a channel restriction. The sum input Hamiltonian and
the uniform local budget give the legal trial state
$\eta=\mathcal R(\rho_m\otimes\rho_n)$ and
\begin{equation}
\label{eq:gluing-variational-chain}
\begin{split}
 E_0((m+n)a)&\leq\operatorname{Tr}(H_{m+n}\eta)\\
 &\leq E_0(ma)+E_0(na)\\
 &\quad+2\varepsilon+K_{B,q,\mathrm{glue}}.
\end{split}
\end{equation}
The constant is independent of $m,n$. Sending
$\varepsilon\downarrow0$ proves Eq.~\eqref{eq:sector-gluing}.
Also, H1--H3 imply at most $\kappa N+\kappa_\partial$ local terms
with norm at most $h_{\max}$. Thus
\begin{equation}
\label{eq:linear-stability}
 \begin{aligned}
 & H_N\geq-h_{\max}(\kappa N+\kappa_\partial)\mathbf1,\\
 & E_0(Na)\geq-AN-D
\end{aligned}
\end{equation}
for length-independent constants. Sector restriction preserves
this lower bound.

\begin{proof}[Proof of Theorem~\ref{thm:bulkexist}]
For a fixed sector and boundary put $u_N=E_0(Na)$.
Equation~\eqref{eq:sector-gluing} gives
$u_{m+n}\leq u_m+u_n+C$ for $m,n\geq n_\ast$.
Set $v_n=u_n+C$. This is subadditive on the same tail.
To apply the subadditive-sequence argument~\cite{Fekete1923}
without using inadmissible short pieces, fix $k\geq n_\ast$ and write
\begin{equation}
\label{eq:tail-remainder}
 n=\ell k+r_0=(\ell-1)k+(k+r_0),\qquad 0\leq r_0<k.
\end{equation}
For large $n$, every block is in the tail, including the final one.
Repeated subadditivity gives
$v_n\leq(\ell-1)v_k+v_{k+r_0}$. At fixed $k$ the last term has only
finitely many possible values, and $(\ell-1)/n\to1/k$.
Therefore $\limsup v_n/n\leq v_k/k$ for every $k\geq n_\ast$.
With $I=\inf_{k\geq n_\ast}v_k/k$,
\begin{equation}
\label{eq:tail-limit}
 I\leq\liminf_{n\to\infty}\frac{v_n}{n}
 \leq\limsup_{n\to\infty}\frac{v_n}{n}\leq I.
\end{equation}
The linear lower bound makes $I>-\infty$, and a fixed finite tail
ratio bounds it above. Since $C/n\to0$, $u_n/n$ has the same finite
limit. The physical path spaces are nonempty on this tail by padding.
Division by fixed $a>0$ gives $e_B^{(s,\alpha)}$.
Apply the argument separately to vacuum and flux at each boundary;
only then subtract their convergent density sequences to obtain
$\sigma_B^{(\alpha)}$.
\end{proof}

\subsection{Two boundary comparisons at the same length}
\label{subsec:boundary-proof}

For $N\geq2r-1$, Table~\ref{tab:four-operations} and the channel
proof give both boundary maps at the same $Na$. They need not be
inverses. Take an open $\varepsilon$-minimizer and, independently,
a periodic one in the canonical domain. The latter approximates
the full-sector infimum by the component and exceptional-block
arguments above. Variational minimization of the two outputs gives
\begin{equation}
\label{eq:two-way-variational-proof}
\begin{aligned}
 E_0^{\rm p}(Na)&\leq E_0^{\rm o}(Na)+\varepsilon
                         +K_{\rm o\to p},\\
 E_0^{\rm o}(Na)&\leq E_0^{\rm p}(Na)+\varepsilon
                         +K_{\rm p\to o}.
\end{aligned}
\end{equation}
Both $K$'s are finite and length independent by
Appendix~\ref{subsec:energy-proof}. Empty retained segments
and overlapping boundary neighborhoods obey the same bounds.
Sending $\varepsilon\downarrow0$ establishes the two independent
energy comparisons; neither reverses the other state transformation.

\begin{proof}[Proof of Theorem~\ref{thm:boundary}]
For each $s$, the comparisons bound
$d_s^B(N)=E_{0,B}^{(s,\mathrm p)}(Na)-E_{0,B}^{(s,\mathrm o)}(Na)$
between $-K_{\rm p\to o,B}^{(s)}$ and $K_{\rm o\to p,B}^{(s)}$.
Their maximum is a valid $C_{\partial,B}^{(s)}$.
The sector density limits already exist by
Theorem~\ref{thm:bulkexist}. Divide the absolute bound by $Na$
and take $N\to\infty$ to equate the two boundary limits for each
sector. Subtracting the vacuum limit from the flux limit proves
$\sigma_B^{\mathrm p}=\sigma_B^{\mathrm o}$.
\end{proof}

At finite $N$, the triangle inequality for
$d_{\rm flux}^B-d_{\rm vac}^B$ gives
Eq.~\eqref{eq:boundary-density-budget}. It controls the
disagreement between boundaries, not the convergence rate of either
density to its own limit.

\end{appendices}

\section*{Supplementary Information}
The Supplementary Information appended to this preprint provides numerical
conventions and qualifications, input and coherence controls, fixed resource
selection, and matched spectra.

\section*{Statements and Declarations}
\subsection*{Funding}
No funding was received for conducting this study.

\subsection*{Competing interests}
The author has no relevant financial or non-financial interests to disclose.

\subsection*{Data availability}
The data that support the findings of this study are available from the
corresponding author upon reasonable request.



\begin{thebibliography}{126}
\ifx \bisbn   \undefined \def \bisbn  #1{ISBN #1}\fi
\ifx \binits  \undefined \def \binits#1{#1}\fi
\ifx \bauthor  \undefined \def \bauthor#1{#1}\fi
\ifx \batitle  \undefined \def \batitle#1{#1}\fi
\ifx \bjtitle  \undefined \def \bjtitle#1{#1}\fi
\ifx \bvolume  \undefined \def \bvolume#1{\textbf{#1}}\fi
\ifx \byear  \undefined \def \byear#1{#1}\fi
\ifx \bissue  \undefined \def \bissue#1{#1}\fi
\ifx \bfpage  \undefined \def \bfpage#1{#1}\fi
\ifx \blpage  \undefined \def \blpage #1{#1}\fi
\ifx \burl  \undefined \def \burl#1{\textsf{#1}}\fi
\ifx \doiurl  \undefined \def \doiurl#1{\url{https://doi.org/#1}}\fi
\ifx \betal  \undefined \def \betal{\textit{et al.}}\fi
\ifx \binstitute  \undefined \def \binstitute#1{#1}\fi
\ifx \binstitutionaled  \undefined \def \binstitutionaled#1{#1}\fi
\ifx \bctitle  \undefined \def \bctitle#1{#1}\fi
\ifx \beditor  \undefined \def \beditor#1{#1}\fi
\ifx \bpublisher  \undefined \def \bpublisher#1{#1}\fi
\ifx \bbtitle  \undefined \def \bbtitle#1{#1}\fi
\ifx \bedition  \undefined \def \bedition#1{#1}\fi
\ifx \bseriesno  \undefined \def \bseriesno#1{#1}\fi
\ifx \blocation  \undefined \def \blocation#1{#1}\fi
\ifx \bsertitle  \undefined \def \bsertitle#1{#1}\fi
\ifx \bsnm \undefined \def \bsnm#1{#1}\fi
\ifx \bsuffix \undefined \def \bsuffix#1{#1}\fi
\ifx \bparticle \undefined \def \bparticle#1{#1}\fi
\ifx \barticle \undefined \def \barticle#1{#1}\fi
\bibcommenthead
\ifx \bconfdate \undefined \def \bconfdate #1{#1}\fi
\ifx \botherref \undefined \def \botherref #1{#1}\fi
\ifx \url \undefined \def \url#1{\textsf{#1}}\fi
\ifx \bchapter \undefined \def \bchapter#1{#1}\fi
\ifx \bbook \undefined \def \bbook#1{#1}\fi
\ifx \bcomment \undefined \def \bcomment#1{#1}\fi
\ifx \oauthor \undefined \def \oauthor#1{#1}\fi
\ifx \citeauthoryear \undefined \def \citeauthoryear#1{#1}\fi
\ifx \endbibitem  \undefined \def \endbibitem {}\fi
\ifx \bconflocation  \undefined \def \bconflocation#1{#1}\fi
\ifx \arxivurl  \undefined \def \arxivurl#1{\textsf{#1}}\fi
\csname PreBibitemsHook\endcsname

\bibitem[\protect\citeauthoryear{{'t Hooft}}{1979}]{tHooft1979}
\begin{barticle}
\bauthor{\bsnm{{'t Hooft}}, \binits{G.}}:
\batitle{{A property of electric and magnetic flux in non-Abelian gauge
  theories}}.
\bjtitle{Nucl. Phys. B}
\bvolume{153},
\bfpage{141}--\blpage{160}
(\byear{1979})
\doiurl{10.1016/0550-3213(79)90595-9}
\end{barticle}
\endbibitem

\bibitem[\protect\citeauthoryear{{Yaffe}}{1980}]{Yaffe1980ConfinementSUN}
\begin{barticle}
\bauthor{\bsnm{{Yaffe}}, \binits{L.G.}}:
\batitle{{Confinement in SU(N) lattice gauge theories}}.
\bjtitle{Phys. Rev. D}
\bvolume{21},
\bfpage{1574}--\blpage{1590}
(\byear{1980})
\doiurl{10.1103/PhysRevD.21.1574}
\end{barticle}
\endbibitem

\bibitem[\protect\citeauthoryear{{Kogut}
  et~al.}{1981}]{KogutPearsonShigemitsu1981}
\begin{barticle}
\bauthor{\bsnm{{Kogut}}, \binits{J.B.}},
\bauthor{\bsnm{{Pearson}}, \binits{R.P.}},
\bauthor{\bsnm{{Shigemitsu}}, \binits{J.}}:
\batitle{{The String Tension, Confinement and Roughening in SU(3) Hamiltonian
  Lattice Gauge Theory}}.
\bjtitle{Phys. Lett. B}
\bvolume{98},
\bfpage{63}--\blpage{68}
(\byear{1981})
\doiurl{10.1016/0370-2693(81)90369-5}
\end{barticle}
\endbibitem

\bibitem[\protect\citeauthoryear{{Kogut}
  et~al.}{1979}]{KogutPearsonShigemitsu1979Beta}
\begin{barticle}
\bauthor{\bsnm{{Kogut}}, \binits{J.B.}},
\bauthor{\bsnm{{Pearson}}, \binits{R.B.}},
\bauthor{\bsnm{{Shigemitsu}}, \binits{J.}}:
\batitle{{Quantum-Chromodynamic $\beta$ Function at Intermediate and Strong
  Coupling}}.
\bjtitle{Phys. Rev. Lett.}
\bvolume{43},
\bfpage{484}--\blpage{486}
(\byear{1979})
\doiurl{10.1103/PhysRevLett.43.484}
\end{barticle}
\endbibitem

\bibitem[\protect\citeauthoryear{{Gaiotto}
  et~al.}{2015}]{GaiottoEtAl2015GeneralizedSymmetries}
\begin{barticle}
\bauthor{\bsnm{{Gaiotto}}, \binits{D.}},
\bauthor{\bsnm{{Kapustin}}, \binits{A.}},
\bauthor{\bsnm{{Seiberg}}, \binits{N.}},
\bauthor{\bsnm{{Willett}}, \binits{B.}}:
\batitle{{Generalized Global Symmetries}}.
\bjtitle{J. High Energy Phys.}
\bvolume{2015}(\bissue{02}),
\bfpage{172}
(\byear{2015})
\doiurl{10.1007/JHEP02(2015)172}
\end{barticle}
\endbibitem

\bibitem[\protect\citeauthoryear{{Hsin}
  et~al.}{2019}]{HsinLamSeiberg2019OneForm}
\begin{barticle}
\bauthor{\bsnm{{Hsin}}, \binits{P.-S.}},
\bauthor{\bsnm{{Lam}}, \binits{H.T.}},
\bauthor{\bsnm{{Seiberg}}, \binits{N.}}:
\batitle{{Comments on One-Form Global Symmetries and Their Gauging in 3d and
  4d}}.
\bjtitle{SciPost Phys.}
\bvolume{6},
\bfpage{039}
(\byear{2019})
\doiurl{10.21468/SciPostPhys.6.3.039}
\end{barticle}
\endbibitem

\bibitem[\protect\citeauthoryear{{Nguyen} et~al.}{2025}]{AddII04}
\begin{barticle}
\bauthor{\bsnm{{Nguyen}}, \binits{M.}},
\bauthor{\bsnm{{Sulejmanpasic}}, \binits{T.}},
\bauthor{\bsnm{{{\"{U}}nsal}}, \binits{M.}}:
\batitle{{Phases of Theories with $\mathbb{Z}_N$ 1-Form Symmetry, and the Roles
  of Center Vortices and Magnetic Monopoles}}.
\bjtitle{Phys. Rev. Lett.}
\bvolume{134},
\bfpage{141902}
(\byear{2025})
\doiurl{10.1103/PhysRevLett.134.141902}
\end{barticle}
\endbibitem

\bibitem[\protect\citeauthoryear{{Mathur} and {Rathor}}{2023}]{AddVII03}
\begin{barticle}
\bauthor{\bsnm{{Mathur}}, \binits{M.}},
\bauthor{\bsnm{{Rathor}}, \binits{A.}}:
\batitle{{Disorder operators and magnetic vortices in SU(N) lattice gauge
  theory}}.
\bjtitle{Phys. Rev. D}
\bvolume{108},
\bfpage{114507}
(\byear{2023})
\doiurl{10.1103/PhysRevD.108.114507}
\end{barticle}
\endbibitem

\bibitem[\protect\citeauthoryear{{de Forcrand} and {von
  Smekal}}{2002}]{AddII20}
\begin{barticle}
\bauthor{\bsnm{{de Forcrand}}, \binits{P.}},
\bauthor{\bsnm{{von Smekal}}, \binits{L.}}:
\batitle{{'t Hooft loops, electric flux sectors, and confinement in SU(2)
  Yang-Mills theory}}.
\bjtitle{Phys. Rev. D}
\bvolume{66},
\bfpage{011504}
(\byear{2002})
\doiurl{10.1103/PhysRevD.66.011504}
\end{barticle}
\endbibitem

\bibitem[\protect\citeauthoryear{{Kogut}
  et~al.}{1976}]{KogutSinclairSusskind1976LowEnergyQCD}
\begin{barticle}
\bauthor{\bsnm{{Kogut}}, \binits{J.}},
\bauthor{\bsnm{{Sinclair}}, \binits{D.K.}},
\bauthor{\bsnm{{Susskind}}, \binits{L.}}:
\batitle{{A quantitative approach to low-energy quantum chromodynamics}}.
\bjtitle{Nucl. Phys. B}
\bvolume{114},
\bfpage{199}--\blpage{236}
(\byear{1976})
\doiurl{10.1016/0550-3213(76)90586-1}
\end{barticle}
\endbibitem

\bibitem[\protect\citeauthoryear{{Byrnes} et~al.}{2004}]{AddVI28}
\begin{barticle}
\bauthor{\bsnm{{Byrnes}}, \binits{T.M.R.}},
\bauthor{\bsnm{{Loan}}, \binits{M.}},
\bauthor{\bsnm{{Hamer}}, \binits{C.J.}},
\bauthor{\bsnm{{Bonnet}}, \binits{F.D.R.}},
\bauthor{\bsnm{{Leinweber}}, \binits{D.B.}},
\bauthor{\bsnm{{Williams}}, \binits{A.G.}},
\bauthor{\bsnm{{Zanotti}}, \binits{J.M.}}:
\batitle{{Hamiltonian limit of (3+1)-dimensional SU(3) lattice gauge theory on
  anisotropic lattices}}.
\bjtitle{Phys. Rev. D}
\bvolume{69},
\bfpage{074509}
(\byear{2004})
\doiurl{10.1103/PhysRevD.69.074509}
\end{barticle}
\endbibitem

\bibitem[\protect\citeauthoryear{{Athenodorou} and {Teper}}{2021}]{AddVI09}
\begin{barticle}
\bauthor{\bsnm{{Athenodorou}}, \binits{A.}},
\bauthor{\bsnm{{Teper}}, \binits{M.}}:
\batitle{{SU(N) gauge theories in 3+1 dimensions: glueball spectrum, string
  tensions and topology}}.
\bjtitle{J. High Energy Phys.}
\bvolume{2021}(\bissue{12}),
\bfpage{082}
(\byear{2021})
\doiurl{10.1007/JHEP12(2021)082}
\end{barticle}
\endbibitem

\bibitem[\protect\citeauthoryear{{Athenodorou} and {Teper}}{2017}]{AddVI16}
\begin{barticle}
\bauthor{\bsnm{{Athenodorou}}, \binits{A.}},
\bauthor{\bsnm{{Teper}}, \binits{M.}}:
\batitle{{SU(N) gauge theories in 2+1 dimensions: glueball spectra and k-string
  tensions}}.
\bjtitle{J. High Energy Phys.}
\bvolume{2017}(\bissue{02}),
\bfpage{015}
(\byear{2017})
\doiurl{10.1007/JHEP02(2017)015}
\end{barticle}
\endbibitem

\bibitem[\protect\citeauthoryear{{Lucini} et~al.}{2004}]{AddVI36}
\begin{barticle}
\bauthor{\bsnm{{Lucini}}, \binits{B.}},
\bauthor{\bsnm{{Teper}}, \binits{M.}},
\bauthor{\bsnm{{Wenger}}, \binits{U.}}:
\batitle{{Glueballs and k-strings in SU(N) gauge theories: Calculations with
  improved operators}}.
\bjtitle{J. High Energy Phys.}
\bvolume{2004}(\bissue{06}),
\bfpage{012}
(\byear{2004})
\doiurl{10.1088/1126-6708/2004/06/012}
\end{barticle}
\endbibitem

\bibitem[\protect\citeauthoryear{{Hasenfratz} et~al.}{1990}]{AddII09}
\begin{barticle}
\bauthor{\bsnm{{Hasenfratz}}, \binits{A.}},
\bauthor{\bsnm{{Hasenfratz}}, \binits{P.}},
\bauthor{\bsnm{{Niedermayer}}, \binits{F.}}:
\batitle{{Electric fluxes and twisted free energies in SU(3)}}.
\bjtitle{Nucl. Phys. B}
\bvolume{329},
\bfpage{739}--\blpage{752}
(\byear{1990})
\doiurl{10.1016/0550-3213(90)90080-W}
\end{barticle}
\endbibitem

\bibitem[\protect\citeauthoryear{{Koller} and {van
  Baal}}{1986}]{KollerVanBaal1986}
\begin{barticle}
\bauthor{\bsnm{{Koller}}, \binits{J.}},
\bauthor{\bsnm{{van Baal}}, \binits{P.}}:
\batitle{{A Rigorous Nonperturbative Result for the Glueball Mass and Electric
  Flux Energy in a Finite Volume}}.
\bjtitle{Nucl. Phys. B}
\bvolume{273},
\bfpage{387}--\blpage{412}
(\byear{1986})
\doiurl{10.1016/0550-3213(86)90252-X}
\end{barticle}
\endbibitem

\bibitem[\protect\citeauthoryear{{L{\"{u}}scher}}{1983}]{Luscher1983}
\begin{barticle}
\bauthor{\bsnm{{L{\"{u}}scher}}, \binits{M.}}:
\batitle{{Some Analytic Results Concerning the Mass Spectrum of Yang--Mills
  Gauge Theories on a Torus}}.
\bjtitle{Nucl. Phys. B}
\bvolume{219},
\bfpage{233}--\blpage{261}
(\byear{1983})
\doiurl{10.1016/0550-3213(83)90436-4}
\end{barticle}
\endbibitem

\bibitem[\protect\citeauthoryear{{Garc{\'{i}}a P{\'{e}}rez}
  et~al.}{2018}]{GarciaPerezEtAl2018}
\begin{barticle}
\bauthor{\bsnm{{Garc{\'{i}}a P{\'{e}}rez}}, \binits{M.}},
\bauthor{\bsnm{{Gonz{\'{a}}lez-Arroyo}}, \binits{A.}},
\bauthor{\bsnm{{Koren}}, \binits{M.}},
\bauthor{\bsnm{{Okawa}}, \binits{M.}}:
\batitle{{The spectrum of 2+1 dimensional Yang--Mills theory on a twisted
  spatial torus}}.
\bjtitle{J. High Energy Phys.}
\bvolume{2018}(\bissue{07}),
\bfpage{169}
(\byear{2018})
\doiurl{10.1007/JHEP07(2018)169}
\end{barticle}
\endbibitem

\bibitem[\protect\citeauthoryear{{van Baal} and {Koller}}{1986}]{AddII21}
\begin{barticle}
\bauthor{\bsnm{{van Baal}}, \binits{P.}},
\bauthor{\bsnm{{Koller}}, \binits{J.}}:
\batitle{{Finite-Size Results for SU(3) Gauge Theory}}.
\bjtitle{Phys. Rev. Lett.}
\bvolume{57},
\bfpage{2783}--\blpage{2786}
(\byear{1986})
\doiurl{10.1103/PhysRevLett.57.2783}
\end{barticle}
\endbibitem

\bibitem[\protect\citeauthoryear{{van Baal}}{1982}]{VanBaal1982Hypertorus}
\begin{barticle}
\bauthor{\bsnm{{van Baal}}, \binits{P.}}:
\batitle{{Some Results for SU(N) Gauge-Fields on the Hypertorus}}.
\bjtitle{Commun. Math. Phys.}
\bvolume{85},
\bfpage{529}--\blpage{547}
(\byear{1982})
\doiurl{10.1007/BF01403503}
\end{barticle}
\endbibitem

\bibitem[\protect\citeauthoryear{{van Baal} and {Koller}}{1987}]{AddII11}
\begin{barticle}
\bauthor{\bsnm{{van Baal}}, \binits{P.}},
\bauthor{\bsnm{{Koller}}, \binits{J.}}:
\batitle{{QCD on a torus, and electric flux energies from tunneling}}.
\bjtitle{Ann. Phys.}
\bvolume{174},
\bfpage{299}--\blpage{371}
(\byear{1987})
\doiurl{10.1016/0003-4916(87)90032-7}
\end{barticle}
\endbibitem

\bibitem[\protect\citeauthoryear{{Meyer}}{2005}]{AddII22}
\begin{barticle}
\bauthor{\bsnm{{Meyer}}, \binits{H.B.}}:
\batitle{{The spectrum of SU(N) gauge theories in finite volume}}.
\bjtitle{J. High Energy Phys.}
\bvolume{2005}(\bissue{03}),
\bfpage{064}
(\byear{2005})
\doiurl{10.1088/1126-6708/2005/03/064}
\end{barticle}
\endbibitem

\bibitem[\protect\citeauthoryear{{Garc{\'{i}}a P{\'{e}}rez}
  et~al.}{2013}]{AddII05}
\begin{barticle}
\bauthor{\bsnm{{Garc{\'{i}}a P{\'{e}}rez}}, \binits{M.}},
\bauthor{\bsnm{{Gonz{\'{a}}lez-Arroyo}}, \binits{A.}},
\bauthor{\bsnm{{Okawa}}, \binits{M.}}:
\batitle{{Spatial volume dependence for 2+1 dimensional SU(N) Yang-Mills
  theory}}.
\bjtitle{J. High Energy Phys.}
\bvolume{2013}(\bissue{09}),
\bfpage{003}
(\byear{2013})
\doiurl{10.1007/JHEP09(2013)003}
\end{barticle}
\endbibitem

\bibitem[\protect\citeauthoryear{{Bergner} et~al.}{2025}]{AddII15}
\begin{barticle}
\bauthor{\bsnm{{Bergner}}, \binits{G.}},
\bauthor{\bsnm{{Gonz{\'{a}}lez-Arroyo}}, \binits{A.}},
\bauthor{\bsnm{{Soler}}, \binits{I.}}:
\batitle{{A $T^2\times R^2$ roadmap to confinement in SU(2) Yang-Mills
  theory}}.
\bjtitle{J. High Energy Phys.}
\bvolume{2025}(\bissue{10}),
\bfpage{087}
(\byear{2025})
\doiurl{10.1007/JHEP10(2025)087}
\end{barticle}
\endbibitem

\bibitem[\protect\citeauthoryear{{Gonz{\'{a}}lez-Arroyo} and
  {Mart{\'{i}}nez}}{1996}]{AddVI18}
\begin{barticle}
\bauthor{\bsnm{{Gonz{\'{a}}lez-Arroyo}}, \binits{A.}},
\bauthor{\bsnm{{Mart{\'{i}}nez}}, \binits{P.}}:
\batitle{{Investigating Yang-Mills theory and confinement as a function of the
  spatial volume}}.
\bjtitle{Nucl. Phys. B}
\bvolume{459},
\bfpage{337}--\blpage{352}
(\byear{1996})
\doiurl{10.1016/0550-3213(95)00601-X}
\end{barticle}
\endbibitem

\bibitem[\protect\citeauthoryear{{Delcamp}
  et~al.}{2016}]{DelcampDittrichRiello2016NonAbelianEntanglement}
\begin{barticle}
\bauthor{\bsnm{{Delcamp}}, \binits{C.}},
\bauthor{\bsnm{{Dittrich}}, \binits{B.}},
\bauthor{\bsnm{{Riello}}, \binits{A.}}:
\batitle{{On entanglement entropy in non-Abelian lattice gauge theory and 3D
  quantum gravity}}.
\bjtitle{J. High Energy Phys.}
\bvolume{2016}(\bissue{11}),
\bfpage{102}
(\byear{2016})
\doiurl{10.1007/JHEP11(2016)102}
\end{barticle}
\endbibitem

\bibitem[\protect\citeauthoryear{{Soni} and
  {Trivedi}}{2016}]{SoniTrivedi2016GaugeEntanglement}
\begin{barticle}
\bauthor{\bsnm{{Soni}}, \binits{R.M.}},
\bauthor{\bsnm{{Trivedi}}, \binits{S.P.}}:
\batitle{{Aspects of entanglement entropy for gauge theories}}.
\bjtitle{J. High Energy Phys.}
\bvolume{2016}(\bissue{01}),
\bfpage{136}
(\byear{2016})
\doiurl{10.1007/JHEP01(2016)136}
\end{barticle}
\endbibitem

\bibitem[\protect\citeauthoryear{{Donnelly}}{2014}]{Donnelly2014NonAbelianEntanglement}
\begin{barticle}
\bauthor{\bsnm{{Donnelly}}, \binits{W.}}:
\batitle{{Entanglement entropy and nonabelian gauge symmetry}}.
\bjtitle{Classical Quantum Gravity}
\bvolume{31},
\bfpage{214003}
(\byear{2014})
\doiurl{10.1088/0264-9381/31/21/214003}
\end{barticle}
\endbibitem

\bibitem[\protect\citeauthoryear{{Aoki}
  et~al.}{2015}]{AokiEtAl2015GaugeEntanglementDefinition}
\begin{barticle}
\bauthor{\bsnm{{Aoki}}, \binits{S.}},
\bauthor{\bsnm{{Iritani}}, \binits{T.}},
\bauthor{\bsnm{{Nozaki}}, \binits{M.}},
\bauthor{\bsnm{{Numasawa}}, \binits{T.}},
\bauthor{\bsnm{{Shiba}}, \binits{N.}},
\bauthor{\bsnm{{Tasaki}}, \binits{H.}}:
\batitle{{On the definition of entanglement entropy in lattice gauge
  theories}}.
\bjtitle{J. High Energy Phys.}
\bvolume{2015}(\bissue{06}),
\bfpage{187}
(\byear{2015})
\doiurl{10.1007/JHEP06(2015)187}
\end{barticle}
\endbibitem

\bibitem[\protect\citeauthoryear{{Buividovich} and
  {Polikarpov}}{2008}]{BuividovichPolikarpov2008GaugeEntanglement}
\begin{barticle}
\bauthor{\bsnm{{Buividovich}}, \binits{P.V.}},
\bauthor{\bsnm{{Polikarpov}}, \binits{M.I.}}:
\batitle{{Entanglement entropy in gauge theories and the holographic principle
  for electric strings}}.
\bjtitle{Phys. Lett. B}
\bvolume{670},
\bfpage{141}--\blpage{145}
(\byear{2008})
\doiurl{10.1016/j.physletb.2008.10.032}
\end{barticle}
\endbibitem

\bibitem[\protect\citeauthoryear{{Lin} and
  {Radi{\v{c}}evi{\'{c}}}}{2020}]{LinRadicevic2020Entanglement}
\begin{barticle}
\bauthor{\bsnm{{Lin}}, \binits{J.}},
\bauthor{\bsnm{{Radi{\v{c}}evi{\'{c}}}}, \binits{{\DJ}.}}:
\batitle{{Comments on defining entanglement entropy}}.
\bjtitle{Nucl. Phys. B}
\bvolume{958},
\bfpage{115118}
(\byear{2020})
\doiurl{10.1016/j.nuclphysb.2020.115118}
\end{barticle}
\endbibitem

\bibitem[\protect\citeauthoryear{{Casini}
  et~al.}{2014}]{CasiniHuertaRosabal2014}
\begin{barticle}
\bauthor{\bsnm{{Casini}}, \binits{H.}},
\bauthor{\bsnm{{Huerta}}, \binits{M.}},
\bauthor{\bsnm{{Rosabal}}, \binits{J.A.}}:
\batitle{{Remarks on entanglement entropy for gauge fields}}.
\bjtitle{Phys. Rev. D}
\bvolume{89},
\bfpage{085012}
(\byear{2014})
\doiurl{10.1103/PhysRevD.89.085012}
\end{barticle}
\endbibitem

\bibitem[\protect\citeauthoryear{{Donnelly}}{2012}]{Donnelly2012BoundaryRepresentations}
\begin{barticle}
\bauthor{\bsnm{{Donnelly}}, \binits{W.}}:
\batitle{{Decomposition of entanglement entropy in lattice gauge theory}}.
\bjtitle{Phys. Rev. D}
\bvolume{85},
\bfpage{085004}
(\byear{2012})
\doiurl{10.1103/PhysRevD.85.085004}
\end{barticle}
\endbibitem

\bibitem[\protect\citeauthoryear{{Hategan-Marandiuc}}{2024}]{HateganMarandiuc2024NonAbelianEntanglement}
\begin{barticle}
\bauthor{\bsnm{{Hategan-Marandiuc}}, \binits{M.}}:
\batitle{{Entanglement entropy in lattices with non-Abelian gauge groups}}.
\bjtitle{Phys. Rev. D}
\bvolume{109},
\bfpage{094501}
(\byear{2024})
\doiurl{10.1103/PhysRevD.109.094501}
\end{barticle}
\endbibitem

\bibitem[\protect\citeauthoryear{{Donnelly} and
  {Freidel}}{2016}]{DonnellyFreidel2016LocalSubsystems}
\begin{barticle}
\bauthor{\bsnm{{Donnelly}}, \binits{W.}},
\bauthor{\bsnm{{Freidel}}, \binits{L.}}:
\batitle{{Local subsystems in gauge theory and gravity}}.
\bjtitle{J. High Energy Phys.}
\bvolume{2016}(\bissue{09}),
\bfpage{102}
(\byear{2016})
\doiurl{10.1007/JHEP09(2016)102}
\end{barticle}
\endbibitem

\bibitem[\protect\citeauthoryear{{Gomes} and
  {Riello}}{2021}]{GomesRiello2021QuasilocalYM}
\begin{barticle}
\bauthor{\bsnm{{Gomes}}, \binits{H.}},
\bauthor{\bsnm{{Riello}}, \binits{A.}}:
\batitle{{The quasilocal degrees of freedom of Yang-Mills theory}}.
\bjtitle{SciPost Phys.}
\bvolume{10},
\bfpage{130}
(\byear{2021})
\doiurl{10.21468/SciPostPhys.10.6.130}
\end{barticle}
\endbibitem

\bibitem[\protect\citeauthoryear{{Riello} and
  {Schiavina}}{2024}]{RielloSchiavina2024FluxSuperselection}
\begin{barticle}
\bauthor{\bsnm{{Riello}}, \binits{A.}},
\bauthor{\bsnm{{Schiavina}}, \binits{M.}}:
\batitle{{Hamiltonian gauge theory with corners: constraint reduction and flux
  superselection}}.
\bjtitle{Adv. Theor. Math. Phys.}
\bvolume{28},
\bfpage{1241}--\blpage{1424}
(\byear{2024})
\doiurl{10.4310/ATMP.241029014101}
\end{barticle}
\endbibitem

\bibitem[\protect\citeauthoryear{{Riello}}{2021}]{Riello2021YMFluxSuperselection}
\begin{barticle}
\bauthor{\bsnm{{Riello}}, \binits{A.}}:
\batitle{{Symplectic reduction of Yang-Mills theory with boundaries: from
  superselection sectors to edge modes, and back}}.
\bjtitle{SciPost Phys.}
\bvolume{10},
\bfpage{125}
(\byear{2021})
\doiurl{10.21468/SciPostPhys.10.6.125}
\end{barticle}
\endbibitem

\bibitem[\protect\citeauthoryear{{Gomes}
  et~al.}{2019}]{GomesHopfmullerRiello2019BoundaryCharges}
\begin{barticle}
\bauthor{\bsnm{{Gomes}}, \binits{H.}},
\bauthor{\bsnm{{Hopfm{\"{u}}ller}}, \binits{F.}},
\bauthor{\bsnm{{Riello}}, \binits{A.}}:
\batitle{{A unified geometric framework for boundary charges and dressings:
  Non-Abelian theory and matter}}.
\bjtitle{Nucl. Phys. B}
\bvolume{941},
\bfpage{249}--\blpage{315}
(\byear{2019})
\doiurl{10.1016/j.nuclphysb.2019.02.020}
\end{barticle}
\endbibitem

\bibitem[\protect\citeauthoryear{{Ball} and
  {Ciambelli}}{2026}]{BallCiambelli2026DynamicalEdgeModes}
\begin{barticle}
\bauthor{\bsnm{{Ball}}, \binits{A.}},
\bauthor{\bsnm{{Ciambelli}}, \binits{L.}}:
\batitle{{Dynamical edge modes in Yang-Mills theory}}.
\bjtitle{SciPost Phys.}
\bvolume{20},
\bfpage{013}
(\byear{2026})
\doiurl{10.21468/SciPostPhys.20.1.013}
\end{barticle}
\endbibitem

\bibitem[\protect\citeauthoryear{{Feldman}
  et~al.}{2024}]{FeldmanEtAl2024SuperselectionEntanglement}
\begin{barticle}
\bauthor{\bsnm{{Feldman}}, \binits{N.}},
\bauthor{\bsnm{{Knaute}}, \binits{J.}},
\bauthor{\bsnm{{Zohar}}, \binits{E.}},
\bauthor{\bsnm{{Goldstein}}, \binits{M.}}:
\batitle{{Superselection-resolved entanglement in lattice gauge theories: a
  tensor network approach}}.
\bjtitle{J. High Energy Phys.}
\bvolume{2024}(\bissue{05}),
\bfpage{083}
(\byear{2024})
\doiurl{10.1007/JHEP05(2024)083}
\end{barticle}
\endbibitem

\bibitem[\protect\citeauthoryear{{Van Acoleyen}
  et~al.}{2016}]{VanAcoleyenEtAl2016DistillationGauge}
\begin{barticle}
\bauthor{\bsnm{{Van Acoleyen}}, \binits{K.}},
\bauthor{\bsnm{{Bultinck}}, \binits{N.}},
\bauthor{\bsnm{{Haegeman}}, \binits{J.}},
\bauthor{\bsnm{{Marien}}, \binits{M.}},
\bauthor{\bsnm{{Scholz}}, \binits{V.B.}},
\bauthor{\bsnm{{Verstraete}}, \binits{F.}}:
\batitle{{Entanglement of Distillation for Lattice Gauge Theories}}.
\bjtitle{Phys. Rev. Lett.}
\bvolume{117},
\bfpage{131602}
(\byear{2016})
\doiurl{10.1103/PhysRevLett.117.131602}
\end{barticle}
\endbibitem

\bibitem[\protect\citeauthoryear{{Geiller} and
  {Jai-akson}}{2020}]{GeillerJaiakson2020EdgeModes}
\begin{barticle}
\bauthor{\bsnm{{Geiller}}, \binits{M.}},
\bauthor{\bsnm{{Jai-akson}}, \binits{P.}}:
\batitle{{Extended actions, dynamics of edge modes, and entanglement entropy}}.
\bjtitle{J. High Energy Phys.}
\bvolume{2020}(\bissue{09}),
\bfpage{134}
(\byear{2020})
\doiurl{10.1007/JHEP09(2020)134}
\end{barticle}
\endbibitem

\bibitem[\protect\citeauthoryear{{Aharony}
  et~al.}{2013}]{AharonySeibergTachikawa2013}
\begingroup
\widowpenalty=10000
\begin{barticle}
\bauthor{\bsnm{{Aharony}}, \binits{O.}},
\bauthor{\bsnm{{Seiberg}}, \binits{N.}},
\bauthor{\bsnm{{Tachikawa}}, \binits{Y.}}:
\batitle{{Reading between the lines of four-dimensional gauge theories}}.
\bjtitle{J. High Energy Phys.}
\bvolume{2013}(\bissue{08}),
\bfpage{115}
(\byear{2013})
\doiurl{10.1007/JHEP08(2013)115}
\end{barticle}
\endbibitem
\par
\endgroup

\bibitem[\protect\citeauthoryear{{Kogut} and
  {Susskind}}{1975}]{KogutSusskind1975}
\begin{barticle}
\bauthor{\bsnm{{Kogut}}, \binits{J.}},
\bauthor{\bsnm{{Susskind}}, \binits{L.}}:
\batitle{{Hamiltonian formulation of Wilson's lattice gauge theories}}.
\bjtitle{Phys. Rev. D}
\bvolume{11},
\bfpage{395}--\blpage{408}
(\byear{1975})
\doiurl{10.1103/PhysRevD.11.395}
\end{barticle}
\endbibitem

\bibitem[\protect\citeauthoryear{{Ciavarella}
  et~al.}{2021}]{CiavarellaKlcoSavage2021}
\begin{barticle}
\bauthor{\bsnm{{Ciavarella}}, \binits{A.}},
\bauthor{\bsnm{{Klco}}, \binits{N.}},
\bauthor{\bsnm{{Savage}}, \binits{M.J.}}:
\batitle{{Trailhead for Quantum Simulation of SU(3) Yang--Mills Lattice Gauge
  Theory in the Local Multiplet Basis}}.
\bjtitle{Phys. Rev. D}
\bvolume{103},
\bfpage{094501}
(\byear{2021})
\doiurl{10.1103/PhysRevD.103.094501}
\end{barticle}
\endbibitem

\bibitem[\protect\citeauthoryear{{Burgio}
  et~al.}{2000}]{BurgioEtAl2000PhysicalHilbert}
\begin{barticle}
\bauthor{\bsnm{{Burgio}}, \binits{G.}},
\bauthor{\bsnm{{De Pietri}}, \binits{R.}},
\bauthor{\bsnm{{Morales-T{\'{e}}cotl}}, \binits{H.A.}},
\bauthor{\bsnm{{Urrutia}}, \binits{L.F.}},
\bauthor{\bsnm{{Vergara}}, \binits{J.D.}}:
\batitle{{The basis of the physical Hilbert space of lattice gauge theories}}.
\bjtitle{Nucl. Phys. B}
\bvolume{566},
\bfpage{547}--\blpage{561}
(\byear{2000})
\doiurl{10.1016/S0550-3213(99)00533-7}
\end{barticle}
\endbibitem

\bibitem[\protect\citeauthoryear{{Anishetty}
  et~al.}{2010}]{AnishettyMathurRaychowdhury2010}
\begin{barticle}
\bauthor{\bsnm{{Anishetty}}, \binits{R.}},
\bauthor{\bsnm{{Mathur}}, \binits{M.}},
\bauthor{\bsnm{{Raychowdhury}}, \binits{I.}}:
\batitle{{Prepotential formulation of SU(3) lattice gauge theory}}.
\bjtitle{J. Phys. A: Math. Theor.}
\bvolume{43},
\bfpage{035403}
(\byear{2010})
\doiurl{10.1088/1751-8113/43/3/035403}
\end{barticle}
\endbibitem

\bibitem[\protect\citeauthoryear{{Ligterink}
  et~al.}{2000}]{LigterinkWaletBishop2000ManyBodySUN}
\begin{barticle}
\bauthor{\bsnm{{Ligterink}}, \binits{N.E.}},
\bauthor{\bsnm{{Walet}}, \binits{N.R.}},
\bauthor{\bsnm{{Bishop}}, \binits{R.F.}}:
\batitle{{Toward a Many-Body Treatment of Hamiltonian Lattice SU(N) Gauge
  Theory}}.
\bjtitle{Ann. Phys.}
\bvolume{284},
\bfpage{215}--\blpage{262}
(\byear{2000})
\doiurl{10.1006/aphy.2000.6070}
\end{barticle}
\endbibitem

\bibitem[\protect\citeauthoryear{{Bronzan} and
  {Vaughan}}{1991}]{BronzanVaughan1991HamiltonianQCDBasis}
\begin{barticle}
\bauthor{\bsnm{{Bronzan}}, \binits{J.B.}},
\bauthor{\bsnm{{Vaughan}}, \binits{T.E.}}:
\batitle{{Basis states for Hamiltonian QCD}}.
\bjtitle{Phys. Rev. D}
\bvolume{44},
\bfpage{3264}--\blpage{3271}
(\byear{1991})
\doiurl{10.1103/PhysRevD.44.3264}
\end{barticle}
\endbibitem

\bibitem[\protect\citeauthoryear{{Mathew} and
  {Raychowdhury}}{2025}]{MathewRaychowdhury2025ProtectingSU3}
\begin{barticle}
\bauthor{\bsnm{{Mathew}}, \binits{E.}},
\bauthor{\bsnm{{Raychowdhury}}, \binits{I.}}:
\batitle{{Protecting gauge symmetries in the dynamics of SU(3) lattice gauge
  theories}}.
\bjtitle{Commun. Phys.}
\bvolume{8},
\bfpage{313}
(\byear{2025})
\doiurl{10.1038/s42005-025-02230-x}
\end{barticle}
\endbibitem

\bibitem[\protect\citeauthoryear{{Pardo}
  et~al.}{2023}]{PardoEtAl2023GaussFermionElimination}
\begin{barticle}
\bauthor{\bsnm{{Pardo}}, \binits{G.}},
\bauthor{\bsnm{{Greenberg}}, \binits{T.}},
\bauthor{\bsnm{{Fortinsky}}, \binits{A.}},
\bauthor{\bsnm{{Katz}}, \binits{N.}},
\bauthor{\bsnm{{Zohar}}, \binits{E.}}:
\batitle{{Resource-efficient quantum simulation of lattice gauge theories in
  arbitrary dimensions: Solving for Gauss's law and fermion elimination}}.
\bjtitle{Phys. Rev. Research}
\bvolume{5},
\bfpage{023077}
(\byear{2023})
\doiurl{10.1103/PhysRevResearch.5.023077}
\end{barticle}
\endbibitem

\bibitem[\protect\citeauthoryear{{de Swart}}{1963}]{deSwart1963SU3CG}
\begin{barticle}
\bauthor{\bsnm{{de Swart}}, \binits{J.J.}}:
\batitle{{The Octet Model and its Clebsch-Gordan Coefficients}}.
\bjtitle{Rev. Mod. Phys.}
\bvolume{35},
\bfpage{916}--\blpage{939}
(\byear{1963})
\doiurl{10.1103/RevModPhys.35.916}
\end{barticle}
\endbibitem

\bibitem[\protect\citeauthoryear{{Kaeding}}{1995}]{Kaeding1995SU3Isoscalar}
\begin{barticle}
\bauthor{\bsnm{{Kaeding}}, \binits{T.A.}}:
\batitle{{Tables of SU(3) Isoscalar Factors}}.
\bjtitle{At. Data Nucl. Data Tables}
\bvolume{61},
\bfpage{233}--\blpage{288}
(\byear{1995})
\doiurl{10.1006/adnd.1995.1011}
\end{barticle}
\endbibitem

\bibitem[\protect\citeauthoryear{{Mukunda} and {Pandit}}{1965}]{AddIII22}
\begin{barticle}
\bauthor{\bsnm{{Mukunda}}, \binits{N.}},
\bauthor{\bsnm{{Pandit}}, \binits{L.K.}}:
\batitle{{Tensor Methods and a Unified Representation Theory of ${SU}_3$}}.
\bjtitle{J. Math. Phys.}
\bvolume{6},
\bfpage{746}--\blpage{765}
(\byear{1965})
\doiurl{10.1063/1.1704332}
\end{barticle}
\endbibitem

\bibitem[\protect\citeauthoryear{{Hecht}}{1965}]{Hecht1965SU3Recoupling}
\begin{barticle}
\bauthor{\bsnm{{Hecht}}, \binits{K.T.}}:
\batitle{{SU3 recoupling and fractional parentage in the 2s-1d shell}}.
\bjtitle{Nucl. Phys.}
\bvolume{62},
\bfpage{1}--\blpage{36}
(\byear{1965})
\doiurl{10.1016/0029-5582(65)90068-4}
\end{barticle}
\endbibitem

\bibitem[\protect\citeauthoryear{{Balaji} et~al.}{2026}]{BalajiEtAl2026}
\begin{barticle}
\bauthor{\bsnm{{Balaji}}, \binits{P.}},
\bauthor{\bsnm{{Conefrey-Shinozaki}}, \binits{C.}},
\bauthor{\bsnm{{Draper}}, \binits{P.}},
\bauthor{\bsnm{{Elhaderi}}, \binits{J.K.}},
\bauthor{\bsnm{{Gupta}}, \binits{D.}},
\bauthor{\bsnm{{Hidalgo}}, \binits{L.}},
\bauthor{\bsnm{{Lytle}}, \binits{A.}}:
\batitle{{Perturbation theory, irrep truncations, and state preparation methods
  for quantum simulations of SU(3) lattice gauge theory}}.
\bjtitle{Phys. Rev. D}
\bvolume{113},
\bfpage{094505}
(\byear{2026})
\doiurl{10.1103/m719-7tdf}
\end{barticle}
\endbibitem

\bibitem[\protect\citeauthoryear{{Chaturvedi} and {Mukunda}}{2002}]{AddIII16}
\begin{barticle}
\bauthor{\bsnm{{Chaturvedi}}, \binits{S.}},
\bauthor{\bsnm{{Mukunda}}, \binits{N.}}:
\batitle{{The Schwinger SU(3) construction. I. Multiplicity problem and
  relation to induced representations}}.
\bjtitle{J. Math. Phys.}
\bvolume{43},
\bfpage{5262}--\blpage{5277}
(\byear{2002})
\doiurl{10.1063/1.1508810}
\end{barticle}
\endbibitem

\bibitem[\protect\citeauthoryear{{Anishetty} et~al.}{2009}]{AddIII15}
\begin{barticle}
\bauthor{\bsnm{{Anishetty}}, \binits{R.}},
\bauthor{\bsnm{{Mathur}}, \binits{M.}},
\bauthor{\bsnm{{Raychowdhury}}, \binits{I.}}:
\batitle{{Irreducible SU(3) Schwinger bosons}}.
\bjtitle{J. Math. Phys.}
\bvolume{50},
\bfpage{053503}
(\byear{2009})
\doiurl{10.1063/1.3122666}
\end{barticle}
\endbibitem

\bibitem[\protect\citeauthoryear{{Anishetty} and
  {Sreeraj}}{2019}]{AnishettySreeraj2019Singlets}
\begin{barticle}
\bauthor{\bsnm{{Anishetty}}, \binits{R.}},
\bauthor{\bsnm{{Sreeraj}}, \binits{T.P.}}:
\batitle{{Addition of SU(3) generators and its singlet Hilbert space}}.
\bjtitle{J. Math. Phys.}
\bvolume{60},
\bfpage{061701}
(\byear{2019})
\doiurl{10.1063/1.5096613}
\end{barticle}
\endbibitem

\bibitem[\protect\citeauthoryear{{Kadam} et~al.}{2025}]{KadamEtAl2025}
\begin{barticle}
\bauthor{\bsnm{{Kadam}}, \binits{S.V.}},
\bauthor{\bsnm{{Naskar}}, \binits{A.}},
\bauthor{\bsnm{{Raychowdhury}}, \binits{I.}},
\bauthor{\bsnm{{Stryker}}, \binits{J.R.}}:
\batitle{{Loop-string-hadron approach to SU(3) lattice Yang--Mills theory:
  Hilbert space of a trivalent vertex}}.
\bjtitle{Phys. Rev. D}
\bvolume{111},
\bfpage{074516}
(\byear{2025})
\doiurl{10.1103/PhysRevD.111.074516}
\end{barticle}
\endbibitem

\bibitem[\protect\citeauthoryear{{Pan} and
  {Draayer}}{1998}]{PanDraayer1998OuterMultiplicityI}
\begin{barticle}
\bauthor{\bsnm{{Pan}}, \binits{F.}},
\bauthor{\bsnm{{Draayer}}, \binits{J.P.}}:
\batitle{{Complementary group resolution of the SU(n) outer multiplicity
  problem. I. The Littlewood rules and a complementary U(2n-2) group
  structure}}.
\bjtitle{J. Math. Phys.}
\bvolume{39},
\bfpage{5631}--\blpage{5641}
(\byear{1998})
\doiurl{10.1063/1.532555}
\end{barticle}
\endbibitem

\bibitem[\protect\citeauthoryear{{Cramp{\'{e}}} et~al.}{2023}]{AddIII04}
\begin{barticle}
\bauthor{\bsnm{{Cramp{\'{e}}}}, \binits{N.}},
\bauthor{\bsnm{{Poulain d'Andecy}}, \binits{L.}},
\bauthor{\bsnm{{Vinet}}, \binits{L.}}:
\batitle{{The Missing Label of $\mathfrak{su}_3$ and Its Symmetry}}.
\bjtitle{Commun. Math. Phys.}
\bvolume{400},
\bfpage{179}--\blpage{213}
(\byear{2023})
\doiurl{10.1007/s00220-022-04596-3}
\end{barticle}
\endbibitem

\bibitem[\protect\citeauthoryear{{Gadiyar} and
  {Sharatchandra}}{1992}]{AddIII17}
\begin{barticle}
\bauthor{\bsnm{{Gadiyar}}, \binits{G.H.}},
\bauthor{\bsnm{{Sharatchandra}}, \binits{H.S.}}:
\batitle{{The missing link: operators for labelling multiplicity in the
  Clebsch-Gordan series}}.
\bjtitle{J. Phys. A: Math. Gen.}
\bvolume{25},
\bfpage{L85}--\blpage{L88}
(\byear{1992})
\doiurl{10.1088/0305-4470/25/3/001}
\end{barticle}
\endbibitem

\bibitem[\protect\citeauthoryear{{Pan} and
  {Draayer}}{1998}]{PanDraayer1998OuterMultiplicityII}
\begin{barticle}
\bauthor{\bsnm{{Pan}}, \binits{F.}},
\bauthor{\bsnm{{Draayer}}, \binits{J.P.}}:
\batitle{{Complementary group resolution of the SU(n) outer multiplicity
  problem. II. Recoupling approach for $SU(3)\supset U(2)$ reduced Wigner
  coefficients}}.
\bjtitle{J. Math. Phys.}
\bvolume{39},
\bfpage{5642}--\blpage{5662}
(\byear{1998})
\doiurl{10.1063/1.532556}
\end{barticle}
\endbibitem

\bibitem[\protect\citeauthoryear{{Burbano} and {Bauer}}{2025}]{AddI03}
\begin{barticle}
\bauthor{\bsnm{{Burbano}}, \binits{I.M.}},
\bauthor{\bsnm{{Bauer}}, \binits{C.W.}}:
\batitle{{Gauge loop-string-hadron formulation on general graphs and
  applications to fully gauge fixed Hamiltonian lattice gauge theory}}.
\bjtitle{J. High Energy Phys.}
\bvolume{2025}(\bissue{12}),
\bfpage{060}
(\byear{2025})
\doiurl{10.1007/jhep12(2025)060}
\end{barticle}
\endbibitem

\bibitem[\protect\citeauthoryear{{Mathur} et~al.}{2010}]{AddIII20}
\begin{barticle}
\bauthor{\bsnm{{Mathur}}, \binits{M.}},
\bauthor{\bsnm{{Raychowdhury}}, \binits{I.}},
\bauthor{\bsnm{{Anishetty}}, \binits{R.}}:
\batitle{{SU(N) irreducible Schwinger bosons}}.
\bjtitle{J. Math. Phys.}
\bvolume{51},
\bfpage{093504}
(\byear{2010})
\doiurl{10.1063/1.3464267}
\end{barticle}
\endbibitem

\bibitem[\protect\citeauthoryear{{Mathur}}{2007}]{AddIII21}
\begin{barticle}
\bauthor{\bsnm{{Mathur}}, \binits{M.}}:
\batitle{{Loop approach to lattice gauge theories}}.
\bjtitle{Nucl. Phys. B}
\bvolume{779},
\bfpage{32}--\blpage{62}
(\byear{2007})
\doiurl{10.1016/j.nuclphysb.2007.04.031}
\end{barticle}
\endbibitem

\bibitem[\protect\citeauthoryear{{Raychowdhury} and {Stryker}}{2020}]{AddIII08}
\begin{barticle}
\bauthor{\bsnm{{Raychowdhury}}, \binits{I.}},
\bauthor{\bsnm{{Stryker}}, \binits{J.R.}}:
\batitle{{Loop, string, and hadron dynamics in SU(2) Hamiltonian lattice gauge
  theories}}.
\bjtitle{Phys. Rev. D}
\bvolume{101},
\bfpage{114502}
(\byear{2020})
\doiurl{10.1103/physrevd.101.114502}
\end{barticle}
\endbibitem

\bibitem[\protect\citeauthoryear{{Davoudi} et~al.}{2021}]{AddIII28}
\begin{barticle}
\bauthor{\bsnm{{Davoudi}}, \binits{Z.}},
\bauthor{\bsnm{{Raychowdhury}}, \binits{I.}},
\bauthor{\bsnm{{Shaw}}, \binits{A.}}:
\batitle{{Search for efficient formulations for Hamiltonian simulation of
  non-Abelian lattice gauge theories}}.
\bjtitle{Phys. Rev. D}
\bvolume{104},
\bfpage{074505}
(\byear{2021})
\doiurl{10.1103/physrevd.104.074505}
\end{barticle}
\endbibitem

\bibitem[\protect\citeauthoryear{{Mathur} and {Rathor}}{2023}]{AddIV11}
\begin{barticle}
\bauthor{\bsnm{{Mathur}}, \binits{M.}},
\bauthor{\bsnm{{Rathor}}, \binits{A.}}:
\batitle{{Exact duality and local dynamics in SU(N) lattice gauge theory}}.
\bjtitle{Phys. Rev. D}
\bvolume{107},
\bfpage{074504}
(\byear{2023})
\doiurl{10.1103/physrevd.107.074504}
\end{barticle}
\endbibitem

\bibitem[\protect\citeauthoryear{{Mariani}}{2024}]{AddIII02}
\begin{barticle}
\bauthor{\bsnm{{Mariani}}, \binits{A.}}:
\batitle{{Almost gauge-invariant states and the ground state of Yang-Mills
  theory}}.
\bjtitle{Phys. Rev. D}
\bvolume{109},
\bfpage{094508}
(\byear{2024})
\doiurl{10.1103/physrevd.109.094508}
\end{barticle}
\endbibitem

\bibitem[\protect\citeauthoryear{{Grabowska} et~al.}{2025}]{AddIII19}
\begin{barticle}
\bauthor{\bsnm{{Grabowska}}, \binits{D.M.}},
\bauthor{\bsnm{{Kane}}, \binits{C.F.}},
\bauthor{\bsnm{{Bauer}}, \binits{C.W.}}:
\batitle{{Fully gauge-fixed SU(2) Hamiltonian for quantum simulations}}.
\bjtitle{Phys. Rev. D}
\bvolume{111},
\bfpage{114516}
(\byear{2025})
\doiurl{10.1103/physrevd.111.114516}
\end{barticle}
\endbibitem

\bibitem[\protect\citeauthoryear{{D'Andrea} et~al.}{2024}]{AddIV14}
\begin{barticle}
\bauthor{\bsnm{{D'Andrea}}, \binits{I.}},
\bauthor{\bsnm{{Bauer}}, \binits{C.W.}},
\bauthor{\bsnm{{Grabowska}}, \binits{D.M.}},
\bauthor{\bsnm{{Freytsis}}, \binits{M.}}:
\batitle{{New basis for Hamiltonian SU(2) simulations}}.
\bjtitle{Phys. Rev. D}
\bvolume{109},
\bfpage{074501}
(\byear{2024})
\doiurl{10.1103/physrevd.109.074501}
\end{barticle}
\endbibitem

\bibitem[\protect\citeauthoryear{{Kaplan} and {Stryker}}{2020}]{AddIII18}
\begin{barticle}
\bauthor{\bsnm{{Kaplan}}, \binits{D.B.}},
\bauthor{\bsnm{{Stryker}}, \binits{J.R.}}:
\batitle{{Gauss's law, duality, and the Hamiltonian formulation of U(1) lattice
  gauge theory}}.
\bjtitle{Phys. Rev. D}
\bvolume{102},
\bfpage{094515}
(\byear{2020})
\doiurl{10.1103/physrevd.102.094515}
\end{barticle}
\endbibitem

\bibitem[\protect\citeauthoryear{{Chandrasekharan} et~al.}{2025}]{AddI01}
\begin{barticle}
\bauthor{\bsnm{{Chandrasekharan}}, \binits{S.}},
\bauthor{\bsnm{{Siew}}, \binits{R.X.}},
\bauthor{\bsnm{{Bhattacharya}}, \binits{T.}}:
\batitle{{Monomer-dimer tensor-network basis for qubit-regularized lattice
  gauge theories}}.
\bjtitle{Phys. Rev. D}
\bvolume{111},
\bfpage{114502}
(\byear{2025})
\doiurl{10.1103/gns9-l3gk}
\end{barticle}
\endbibitem

\bibitem[\protect\citeauthoryear{{Hayata} and
  {Hidaka}}{2023}]{HayataHidaka2023QDeformed}
\begin{barticle}
\bauthor{\bsnm{{Hayata}}, \binits{T.}},
\bauthor{\bsnm{{Hidaka}}, \binits{Y.}}:
\batitle{{q deformed formulation of Hamiltonian SU(3) Yang--Mills theory}}.
\bjtitle{J. High Energy Phys.}
\bvolume{2023}(\bissue{09}),
\bfpage{123}
(\byear{2023})
\doiurl{10.1007/JHEP09(2023)123}
\end{barticle}
\endbibitem

\bibitem[\protect\citeauthoryear{{Zache} et~al.}{2023}]{AddIII25}
\begin{barticle}
\bauthor{\bsnm{{Zache}}, \binits{T.V.}},
\bauthor{\bsnm{{Gonz{\'{a}}lez-Cuadra}}, \binits{D.}},
\bauthor{\bsnm{{Zoller}}, \binits{P.}}:
\batitle{{Quantum and Classical Spin-Network Algorithms for q-Deformed
  Kogut-Susskind Gauge Theories}}.
\bjtitle{Phys. Rev. Lett.}
\bvolume{131},
\bfpage{171902}
(\byear{2023})
\doiurl{10.1103/physrevlett.131.171902}
\end{barticle}
\endbibitem

\bibitem[\protect\citeauthoryear{{Ciavarella}}{2023}]{AddIV07}
\begin{barticle}
\bauthor{\bsnm{{Ciavarella}}, \binits{A.N.}}:
\batitle{{Quantum simulation of lattice QCD with improved Hamiltonians}}.
\bjtitle{Phys. Rev. D}
\bvolume{108},
\bfpage{094513}
(\byear{2023})
\doiurl{10.1103/physrevd.108.094513}
\end{barticle}
\endbibitem

\bibitem[\protect\citeauthoryear{{Fontana} et~al.}{2025}]{AddIII13}
\begin{barticle}
\bauthor{\bsnm{{Fontana}}, \binits{P.}},
\bauthor{\bsnm{{Miranda-Riaza}}, \binits{M.}},
\bauthor{\bsnm{{Celi}}, \binits{A.}}:
\batitle{{Efficient Finite-Resource Formulation of Non-Abelian Lattice Gauge
  Theories beyond One Dimension}}.
\bjtitle{Phys. Rev. X}
\bvolume{15},
\bfpage{031065}
(\byear{2025})
\doiurl{10.1103/k9p6-c649}
\end{barticle}
\endbibitem

\bibitem[\protect\citeauthoryear{{Zohar} and
  {Burrello}}{2015}]{ZoharBurrello2015}
\begin{barticle}
\bauthor{\bsnm{{Zohar}}, \binits{E.}},
\bauthor{\bsnm{{Burrello}}, \binits{M.}}:
\batitle{{Formulation of lattice gauge theories for quantum simulations}}.
\bjtitle{Phys. Rev. D}
\bvolume{91},
\bfpage{054506}
(\byear{2015})
\doiurl{10.1103/PhysRevD.91.054506}
\end{barticle}
\endbibitem

\bibitem[\protect\citeauthoryear{{Hayata} and {Hidaka}}{2023}]{AddIII06}
\begin{barticle}
\bauthor{\bsnm{{Hayata}}, \binits{T.}},
\bauthor{\bsnm{{Hidaka}}, \binits{Y.}}:
\batitle{{String-net formulation of Hamiltonian lattice Yang-Mills theories and
  quantum many-body scars in a nonabelian gauge theory}}.
\bjtitle{J. High Energy Phys.}
\bvolume{2023}(\bissue{09}),
\bfpage{126}
(\byear{2023})
\doiurl{10.1007/jhep09(2023)126}
\end{barticle}
\endbibitem

\bibitem[\protect\citeauthoryear{{Hayata} et~al.}{2026}]{AddIV09}
\begin{barticle}
\bauthor{\bsnm{{Hayata}}, \binits{T.}},
\bauthor{\bsnm{{Hidaka}}, \binits{Y.}},
\bauthor{\bsnm{{Watanabe}}, \binits{H.}}:
\batitle{{Phases of the q-deformed SU(N) Yang-Mills theory at large N}}.
\bjtitle{Phys. Rev. D}
\bvolume{113},
\bfpage{074517}
(\byear{2026})
\doiurl{10.1103/ltwh-42sk}
\end{barticle}
\endbibitem

\bibitem[\protect\citeauthoryear{{Bonzom} et~al.}{2023}]{AddVIII02}
\begin{barticle}
\bauthor{\bsnm{{Bonzom}}, \binits{V.}},
\bauthor{\bsnm{{Dupuis}}, \binits{M.}},
\bauthor{\bsnm{{Girelli}}, \binits{F.}},
\bauthor{\bsnm{{Pan}}, \binits{Q.}}:
\batitle{{Local observables in $\mathrm{SU}_q(2)$ lattice gauge theory}}.
\bjtitle{Phys. Rev. D}
\bvolume{107},
\bfpage{026014}
(\byear{2023})
\doiurl{10.1103/PhysRevD.107.026014}
\end{barticle}
\endbibitem

\bibitem[\protect\citeauthoryear{{Illa} et~al.}{2025}]{AddIII14}
\begin{barticle}
\bauthor{\bsnm{{Illa}}, \binits{M.}},
\bauthor{\bsnm{{Savage}}, \binits{M.J.}},
\bauthor{\bsnm{{Yao}}, \binits{X.}}:
\batitle{{Improved honeycomb and hyperhoneycomb lattice Hamiltonians for
  quantum simulations of non-Abelian gauge theories}}.
\bjtitle{Phys. Rev. D}
\bvolume{111},
\bfpage{114520}
(\byear{2025})
\doiurl{10.1103/3rwf-f844}
\end{barticle}
\endbibitem

\bibitem[\protect\citeauthoryear{{Brower}
  et~al.}{1999}]{BrowerChandrasekharanWiese1999}
\begin{barticle}
\bauthor{\bsnm{{Brower}}, \binits{R.}},
\bauthor{\bsnm{{Chandrasekharan}}, \binits{S.}},
\bauthor{\bsnm{{Wiese}}, \binits{U.-J.}}:
\batitle{{QCD as a quantum link model}}.
\bjtitle{Phys. Rev. D}
\bvolume{60},
\bfpage{094502}
(\byear{1999})
\doiurl{10.1103/PhysRevD.60.094502}
\end{barticle}
\endbibitem

\bibitem[\protect\citeauthoryear{{Chandrasekharan} and {Wiese}}{1997}]{AddIV08}
\begin{barticle}
\bauthor{\bsnm{{Chandrasekharan}}, \binits{S.}},
\bauthor{\bsnm{{Wiese}}, \binits{U.-J.}}:
\batitle{{Quantum link models: A discrete approach to gauge theories}}.
\bjtitle{Nucl. Phys. B}
\bvolume{492},
\bfpage{455}--\blpage{471}
(\byear{1997})
\doiurl{10.1016/s0550-3213(97)80041-7}
\end{barticle}
\endbibitem

\bibitem[\protect\citeauthoryear{{Pradhan} et~al.}{2024}]{AddII06}
\begin{barticle}
\bauthor{\bsnm{{Pradhan}}, \binits{S.}},
\bauthor{\bsnm{{Maroncelli}}, \binits{A.}},
\bauthor{\bsnm{{Ercolessi}}, \binits{E.}}:
\batitle{{Discrete Abelian lattice gauge theories on a ladder and their
  dualities with quantum clock models}}.
\bjtitle{Phys. Rev. B}
\bvolume{109},
\bfpage{064410}
(\byear{2024})
\doiurl{10.1103/PhysRevB.109.064410}
\end{barticle}
\endbibitem

\bibitem[\protect\citeauthoryear{{Ciavarella} et~al.}{2025}]{AddI04}
\begin{barticle}
\bauthor{\bsnm{{Ciavarella}}, \binits{A.N.}},
\bauthor{\bsnm{{Burbano}}, \binits{I.M.}},
\bauthor{\bsnm{{Bauer}}, \binits{C.W.}}:
\batitle{{Efficient truncations of $SU(N_c)$ lattice gauge theory for quantum
  simulation}}.
\bjtitle{Phys. Rev. D}
\bvolume{112},
\bfpage{054514}
(\byear{2025})
\doiurl{10.1103/ylqb-phv5}
\end{barticle}
\endbibitem

\bibitem[\protect\citeauthoryear{Ciavarella et~al.}{2026}]{AddI05}
\begin{barticle}
\bauthor{\bsnm{Ciavarella}, \binits{A.N.}},
\bauthor{\bsnm{Hariprakash}, \binits{S.}},
\bauthor{\bsnm{Halimeh}, \binits{J.C.}},
\bauthor{\bsnm{Bauer}, \binits{C.W.}}:
\batitle{{Truncation uncertainties for accurate quantum simulations of lattice
  gauge theories}}.
\bjtitle{Quantum}
\bvolume{10},
\bfpage{2216}
(\byear{2026})
\doiurl{10.22331/q-2026-09-24-2216}
{\href{https://arxiv.org/abs/2508.00061}{{arXiv:2508.00061}}}
{[quant-ph]}
\end{barticle}
\endbibitem

\bibitem[\protect\citeauthoryear{{Liu} and
  {Chandrasekharan}}{2022}]{LiuChandrasekharan2022QubitRegularization}
\begin{barticle}
\bauthor{\bsnm{{Liu}}, \binits{H.}},
\bauthor{\bsnm{{Chandrasekharan}}, \binits{S.}}:
\batitle{{Qubit Regularization and Qubit Embedding Algebras}}.
\bjtitle{Symmetry}
\bvolume{14},
\bfpage{305}
(\byear{2022})
\doiurl{10.3390/sym14020305}
\end{barticle}
\endbibitem

\bibitem[\protect\citeauthoryear{{Ciavarella} and {Bauer}}{2024}]{AddIV02}
\begin{barticle}
\bauthor{\bsnm{{Ciavarella}}, \binits{A.N.}},
\bauthor{\bsnm{{Bauer}}, \binits{C.W.}}:
\batitle{{Quantum Simulation of SU(3) Lattice Yang-Mills Theory at Leading
  Order in Large-$N_c$ Expansion}}.
\bjtitle{Phys. Rev. Lett.}
\bvolume{133},
\bfpage{111901}
(\byear{2024})
\doiurl{10.1103/physrevlett.133.111901}
\end{barticle}
\endbibitem

\bibitem[\protect\citeauthoryear{{Assi} and {Lamm}}{2024}]{AddIV16}
\begin{barticle}
\bauthor{\bsnm{{Assi}}, \binits{B.}},
\bauthor{\bsnm{{Lamm}}, \binits{H.}}:
\batitle{{Digitization and subduction of SU(N) gauge theories}}.
\bjtitle{Phys. Rev. D}
\bvolume{110},
\bfpage{074511}
(\byear{2024})
\doiurl{10.1103/physrevd.110.074511}
\end{barticle}
\endbibitem

\bibitem[\protect\citeauthoryear{{Romiti} and {Urbach}}{2024}]{AddIV15}
\begin{barticle}
\bauthor{\bsnm{{Romiti}}, \binits{S.}},
\bauthor{\bsnm{{Urbach}}, \binits{C.}}:
\batitle{{Digitizing lattice gauge theories in the magnetic basis: reducing the
  breaking of the fundamental commutation relations}}.
\bjtitle{Eur. Phys. J. C}
\bvolume{84},
\bfpage{708}
(\byear{2024})
\doiurl{10.1140/epjc/s10052-024-13037-5}
\end{barticle}
\endbibitem

\bibitem[\protect\citeauthoryear{{Carena}
  et~al.}{2022}]{CarenaEtAl2022ImprovedHamiltonians}
\begin{barticle}
\bauthor{\bsnm{{Carena}}, \binits{M.}},
\bauthor{\bsnm{{Lamm}}, \binits{H.}},
\bauthor{\bsnm{{Li}}, \binits{Y.-Y.}},
\bauthor{\bsnm{{Liu}}, \binits{W.}}:
\batitle{{Improved Hamiltonians for Quantum Simulations of Gauge Theories}}.
\bjtitle{Phys. Rev. Lett.}
\bvolume{129},
\bfpage{051601}
(\byear{2022})
\doiurl{10.1103/PhysRevLett.129.051601}
\end{barticle}
\endbibitem

\bibitem[\protect\citeauthoryear{{Tong}
  et~al.}{2022}]{TongEtAl2022ProvablyAccurate}
\begin{barticle}
\bauthor{\bsnm{{Tong}}, \binits{Y.}},
\bauthor{\bsnm{{Albert}}, \binits{V.V.}},
\bauthor{\bsnm{{McClean}}, \binits{J.R.}},
\bauthor{\bsnm{{Preskill}}, \binits{J.}},
\bauthor{\bsnm{{Su}}, \binits{Y.}}:
\batitle{{Provably accurate simulation of gauge theories and bosonic systems}}.
\bjtitle{Quantum}
\bvolume{6},
\bfpage{816}
(\byear{2022})
\doiurl{10.22331/q-2022-09-22-816}
\end{barticle}
\endbibitem

\bibitem[\protect\citeauthoryear{{Ciavarella}
  et~al.}{2025}]{CiavarellaBauerHalimeh2025Fragmentation}
\begin{barticle}
\bauthor{\bsnm{{Ciavarella}}, \binits{A.N.}},
\bauthor{\bsnm{{Bauer}}, \binits{C.W.}},
\bauthor{\bsnm{{Halimeh}}, \binits{J.C.}}:
\batitle{{Generic Hilbert space fragmentation in Kogut--Susskind lattice gauge
  theories}}.
\bjtitle{Phys. Rev. D}
\bvolume{112},
\bfpage{L091501}
(\byear{2025})
\doiurl{10.1103/p4xd-4ngx}
\end{barticle}
\endbibitem

\bibitem[\protect\citeauthoryear{{Klinger} and
  {Leigh}}{2024}]{KlingerLeigh2024ConditionalExpectations}
\begin{barticle}
\bauthor{\bsnm{{Klinger}}, \binits{M.S.}},
\bauthor{\bsnm{{Leigh}}, \binits{R.G.}}:
\batitle{{Crossed products, conditional expectations and constraint
  quantization}}.
\bjtitle{Nucl. Phys. B}
\bvolume{1006},
\bfpage{116622}
(\byear{2024})
\doiurl{10.1016/j.nuclphysb.2024.116622}
\end{barticle}
\endbibitem

\bibitem[\protect\citeauthoryear{{Hollenberg} and
  {Witte}}{1994}]{HollenbergWitte1994EnergyDensity}
\begin{barticle}
\bauthor{\bsnm{{Hollenberg}}, \binits{L.C.L.}},
\bauthor{\bsnm{{Witte}}, \binits{N.S.}}:
\batitle{{General nonperturbative estimate of the energy density of lattice
  Hamiltonians}}.
\bjtitle{Phys. Rev. D}
\bvolume{50},
\bfpage{3382}--\blpage{3386}
(\byear{1994})
\doiurl{10.1103/PhysRevD.50.3382}
\end{barticle}
\endbibitem

\bibitem[\protect\citeauthoryear{{Bronzan} and
  {G{\"{u}}nal}}{1994}]{BronzanGunal1994PureSU3LargeLattice}
\begin{barticle}
\bauthor{\bsnm{{Bronzan}}, \binits{J.B.}},
\bauthor{\bsnm{{G{\"{u}}nal}}, \binits{Y.}}:
\batitle{{Solving Schr{\"{o}}dinger's equation for pure SU(3) on a large
  lattice}}.
\bjtitle{Phys. Rev. D}
\bvolume{50},
\bfpage{5924}--\blpage{5934}
(\byear{1994})
\doiurl{10.1103/PhysRevD.50.5924}
\end{barticle}
\endbibitem

\bibitem[\protect\citeauthoryear{{Grundling} and {Rudolph}}{2017}]{AddIV05}
\begin{barticle}
\bauthor{\bsnm{{Grundling}}, \binits{H.}},
\bauthor{\bsnm{{Rudolph}}, \binits{G.}}:
\batitle{{Dynamics for QCD on an Infinite Lattice}}.
\bjtitle{Commun. Math. Phys.}
\bvolume{349},
\bfpage{1163}--\blpage{1202}
(\byear{2017})
\doiurl{10.1007/s00220-016-2733-5}
\end{barticle}
\endbibitem

\bibitem[\protect\citeauthoryear{{Grundling} and {Rudolph}}{2013}]{AddIV06}
\begin{barticle}
\bauthor{\bsnm{{Grundling}}, \binits{H.}},
\bauthor{\bsnm{{Rudolph}}, \binits{G.}}:
\batitle{{QCD on an Infinite Lattice}}.
\bjtitle{Commun. Math. Phys.}
\bvolume{318},
\bfpage{717}--\blpage{766}
(\byear{2013})
\doiurl{10.1007/s00220-013-1674-5}
\end{barticle}
\endbibitem

\bibitem[\protect\citeauthoryear{{Balachandran} et~al.}{2022}]{AddII08}
\begin{barticle}
\bauthor{\bsnm{{Balachandran}}, \binits{A.P.}},
\bauthor{\bsnm{{Nair}}, \binits{V.P.}},
\bauthor{\bsnm{{Pinzul}}, \binits{A.}},
\bauthor{\bsnm{{Reyes-Lega}}, \binits{A.F.}},
\bauthor{\bsnm{{Vaidya}}, \binits{S.}}:
\batitle{{Superselection, boundary algebras, and duality in gauge theories}}.
\bjtitle{Phys. Rev. D}
\bvolume{106},
\bfpage{025001}
(\byear{2022})
\doiurl{10.1103/PhysRevD.106.025001}
\end{barticle}
\endbibitem

\bibitem[\protect\citeauthoryear{{Kijowski} and {Rudolph}}{2005}]{AddII12}
\begin{barticle}
\bauthor{\bsnm{{Kijowski}}, \binits{J.}},
\bauthor{\bsnm{{Rudolph}}, \binits{G.}}:
\batitle{{Charge superselection sectors for QCD on the lattice}}.
\bjtitle{J. Math. Phys.}
\bvolume{46},
\bfpage{032303}
(\byear{2005})
\doiurl{10.1063/1.1851604}
\end{barticle}
\endbibitem

\bibitem[\protect\citeauthoryear{{Kijowski} and {Rudolph}}{2002}]{AddII23}
\begin{barticle}
\bauthor{\bsnm{{Kijowski}}, \binits{J.}},
\bauthor{\bsnm{{Rudolph}}, \binits{G.}}:
\batitle{{On the Gauss law and global charge for quantum chromodynamics}}.
\bjtitle{J. Math. Phys.}
\bvolume{43},
\bfpage{1796}--\blpage{1808}
(\byear{2002})
\doiurl{10.1063/1.1447310}
\end{barticle}
\endbibitem

\bibitem[\protect\citeauthoryear{{Balachandran} et~al.}{2019}]{AddII19}
\begin{barticle}
\bauthor{\bsnm{{Balachandran}}, \binits{A.P.}},
\bauthor{\bsnm{{Nair}}, \binits{V.P.}},
\bauthor{\bsnm{{Vaidya}}, \binits{S.}}:
\batitle{{Aspects of boundary conditions for non-Abelian gauge theories}}.
\bjtitle{Phys. Rev. D}
\bvolume{100},
\bfpage{045001}
(\byear{2019})
\doiurl{10.1103/PhysRevD.100.045001}
\end{barticle}
\endbibitem

\bibitem[\protect\citeauthoryear{{Cohen}}{2025}]{AddVII02}
\begin{barticle}
\bauthor{\bsnm{{Cohen}}, \binits{T.D.}}:
\batitle{{Gauge invariance and color charge fluctuations}}.
\bjtitle{Phys. Rev. C}
\bvolume{112},
\bfpage{L042201}
(\byear{2025})
\doiurl{10.1103/czq9-xl61}
\end{barticle}
\endbibitem

\bibitem[\protect\citeauthoryear{{Henheik}
  et~al.}{2022}]{HenheikTeufelWessel2022LocalStability}
\begin{barticle}
\bauthor{\bsnm{{Henheik}}, \binits{J.}},
\bauthor{\bsnm{{Teufel}}, \binits{S.}},
\bauthor{\bsnm{{Wessel}}, \binits{T.}}:
\batitle{{Local stability of ground states in locally gapped and weakly
  interacting quantum spin systems}}.
\bjtitle{Lett. Math. Phys.}
\bvolume{112},
\bfpage{9}
(\byear{2022})
\doiurl{10.1007/s11005-021-01494-y}
\end{barticle}
\endbibitem

\bibitem[\protect\citeauthoryear{{Kanazawa}}{2009}]{Kanazawa2009TomboulisYaffeSUN}
\begin{barticle}
\bauthor{\bsnm{{Kanazawa}}, \binits{T.}}:
\batitle{{Generalizing the Tomboulis--Yaffe inequality to SU(N) lattice gauge
  theories and general classical spin systems}}.
\bjtitle{Ann. Phys.}
\bvolume{324},
\bfpage{1634}--\blpage{1665}
(\byear{2009})
\doiurl{10.1016/j.aop.2009.04.002}
\end{barticle}
\endbibitem

\bibitem[\protect\citeauthoryear{{Morikawa} and {Suzuki}}{2025}]{AddII01}
\begin{barticle}
\bauthor{\bsnm{{Morikawa}}, \binits{O.}},
\bauthor{\bsnm{{Suzuki}}, \binits{H.}}:
\batitle{{Direct Monte Carlo Computation of the 't~Hooft Partition Function}}.
\bjtitle{Prog. Theor. Exp. Phys.}
\bvolume{2025},
\bfpage{063B04}
(\byear{2025})
\doiurl{10.1093/ptep/ptaf072}
\end{barticle}
\endbibitem

\bibitem[\protect\citeauthoryear{{Maeda} and {Tanizaki}}{2025}]{AddII07}
\begin{barticle}
\bauthor{\bsnm{{Maeda}}, \binits{J.}},
\bauthor{\bsnm{{Tanizaki}}, \binits{Y.}}:
\batitle{{Twisted partition functions as order parameters}}.
\bjtitle{J. High Energy Phys.}
\bvolume{2025}(\bissue{08}),
\bfpage{128}
(\byear{2025})
\doiurl{10.1007/JHEP08(2025)128}
\end{barticle}
\endbibitem

\bibitem[\protect\citeauthoryear{Chen et~al.}{2026}]{ChenMullerYao2026}
\begin{barticle}
\bauthor{\bsnm{Chen}, \binits{V.}},
\bauthor{\bsnm{M{\"u}ller}, \binits{B.}},
\bauthor{\bsnm{Yao}, \binits{X.}}:
\batitle{{Minimally Truncated SU(3) Lattice Gauge Theory and String Tension}}.
\bjtitle{Phys. Rev. D}
(\byear{2026})
\doiurl{10.1103/1qt7-3ywl}
{\href{https://arxiv.org/abs/2601.10065}{{arXiv:2601.10065}}}
{[hep-lat]}.
\bcomment{Accepted for publication, 11 September 2026}
\end{barticle}
\endbibitem

\bibitem[\protect\citeauthoryear{{Chin}
  et~al.}{1988}]{ChinLongRobson1988SU3GroundState}
\begin{barticle}
\bauthor{\bsnm{{Chin}}, \binits{S.A.}},
\bauthor{\bsnm{{Long}}, \binits{C.}},
\bauthor{\bsnm{{Robson}}, \binits{D.}}:
\batitle{{Exact ground-state properties of SU(3) Hamiltonian lattice gauge
  theory}}.
\bjtitle{Phys. Rev. D}
\bvolume{37},
\bfpage{3001}--\blpage{3005}
(\byear{1988})
\doiurl{10.1103/PhysRevD.37.3001}
\end{barticle}
\endbibitem

\bibitem[\protect\citeauthoryear{{Hamer}
  et~al.}{2000}]{HamerSamarasBursill2000SU3GFMC}
\begin{barticle}
\bauthor{\bsnm{{Hamer}}, \binits{C.J.}},
\bauthor{\bsnm{{Samaras}}, \binits{M.}},
\bauthor{\bsnm{{Bursill}}, \binits{R.J.}}:
\batitle{{Green's Function Monte Carlo study of SU(3) lattice gauge theory in
  (3+1)D}}.
\bjtitle{Phys. Rev. D}
\bvolume{62},
\bfpage{074506}
(\byear{2000})
\doiurl{10.1103/PhysRevD.62.074506}
\end{barticle}
\endbibitem

\bibitem[\protect\citeauthoryear{{Buyens} et~al.}{2017}]{AddIV03}
\begin{barticle}
\bauthor{\bsnm{{Buyens}}, \binits{B.}},
\bauthor{\bsnm{{Montangero}}, \binits{S.}},
\bauthor{\bsnm{{Haegeman}}, \binits{J.}},
\bauthor{\bsnm{{Verstraete}}, \binits{F.}},
\bauthor{\bsnm{{Van Acoleyen}}, \binits{K.}}:
\batitle{{Finite-representation approximation of lattice gauge theories at the
  continuum limit with tensor networks}}.
\bjtitle{Phys. Rev. D}
\bvolume{95},
\bfpage{094509}
(\byear{2017})
\doiurl{10.1103/physrevd.95.094509}
\end{barticle}
\endbibitem

\bibitem[\protect\citeauthoryear{{DeGrand}
  et~al.}{1996}]{DeGrandEtAl1996FixedPointSU3}
\begin{barticle}
\bauthor{\bsnm{{DeGrand}}, \binits{T.}},
\bauthor{\bsnm{{Hasenfratz}}, \binits{A.}},
\bauthor{\bsnm{{Hasenfratz}}, \binits{P.}},
\bauthor{\bsnm{{Niedermayer}}, \binits{F.}}:
\batitle{{Fixed point actions for SU(3) gauge theory}}.
\bjtitle{Phys. Lett. B}
\bvolume{365},
\bfpage{233}--\blpage{238}
(\byear{1996})
\doiurl{10.1016/0370-2693(95)01233-8}
\end{barticle}
\endbibitem

\bibitem[\protect\citeauthoryear{{Chin}
  et~al.}{1986}]{ChinLongRobson1986GlueballScaling}
\begin{barticle}
\bauthor{\bsnm{{Chin}}, \binits{S.A.}},
\bauthor{\bsnm{{Long}}, \binits{C.}},
\bauthor{\bsnm{{Robson}}, \binits{D.}}:
\batitle{{Scaling of the $0^{++}$ Glueball Mass in SU(N) Hamiltonian Lattice
  Calculations}}.
\bjtitle{Phys. Rev. Lett.}
\bvolume{57},
\bfpage{2779}--\blpage{2782}
(\byear{1986})
\doiurl{10.1103/PhysRevLett.57.2779}
\end{barticle}
\endbibitem

\bibitem[\protect\citeauthoryear{{Alexandru}
  et~al.}{2022}]{AlexandruBedaqueBrettLamm2022DigitizedQCD}
\begin{barticle}
\bauthor{\bsnm{{Alexandru}}, \binits{A.}},
\bauthor{\bsnm{{Bedaque}}, \binits{P.F.}},
\bauthor{\bsnm{{Brett}}, \binits{R.}},
\bauthor{\bsnm{{Lamm}}, \binits{H.}}:
\batitle{{Spectrum of digitized QCD: Glueballs in a $S(1080)$ gauge theory}}.
\bjtitle{Phys. Rev. D}
\bvolume{105},
\bfpage{114508}
(\byear{2022})
\doiurl{10.1103/PhysRevD.105.114508}
\end{barticle}
\endbibitem

\bibitem[\protect\citeauthoryear{{Siew}
  et~al.}{2026}]{SiewChandrasekharanBhattacharya2026SU2Qubit}
\begin{barticle}
\bauthor{\bsnm{{Siew}}, \binits{R.X.}},
\bauthor{\bsnm{{Chandrasekharan}}, \binits{S.}},
\bauthor{\bsnm{{Bhattacharya}}, \binits{T.}}:
\batitle{{Asymptotic freedom and massive glueballs in a qubit-regularized SU(2)
  gauge theory}}.
\bjtitle{Phys. Rev. D}
\bvolume{113},
\bfpage{114522}
(\byear{2026})
\doiurl{10.1103/2cy4-cp7f}
\end{barticle}
\endbibitem

\bibitem[\protect\citeauthoryear{{Jakobs}
  et~al.}{2025}]{JakobsEtAl2025PartitioningsSU2}
\begin{barticle}
\bauthor{\bsnm{{Jakobs}}, \binits{T.}},
\bauthor{\bsnm{{Garofalo}}, \binits{M.}},
\bauthor{\bsnm{{Hartung}}, \binits{T.}},
\bauthor{\bsnm{{Jansen}}, \binits{K.}},
\bauthor{\bsnm{{Ostmeyer}}, \binits{J.}},
\bauthor{\bsnm{{Romiti}}, \binits{S.}},
\bauthor{\bsnm{{Urbach}}, \binits{C.}}:
\batitle{{Dynamics in hamiltonian lattice gauge theory: approaching the
  continuum limit with partitionings of SU(2)}}.
\bjtitle{Eur. Phys. J. C}
\bvolume{85},
\bfpage{1418}
(\byear{2025})
\doiurl{10.1140/epjc/s10052-025-15120-x}
\end{barticle}
\endbibitem

\bibitem[\protect\citeauthoryear{{Sharifian} et~al.}{2023}]{AddVI24}
\begin{barticle}
\bauthor{\bsnm{{Sharifian}}, \binits{A.}},
\bauthor{\bsnm{{Cardoso}}, \binits{N.}},
\bauthor{\bsnm{{Bicudo}}, \binits{P.}}:
\batitle{{Eight very excited flux tube spectra and possible axions in SU(3)
  lattice gauge theory}}.
\bjtitle{Phys. Rev. D}
\bvolume{107},
\bfpage{114507}
(\byear{2023})
\doiurl{10.1103/PhysRevD.107.114507}
\end{barticle}
\endbibitem

\bibitem[\protect\citeauthoryear{{Bicudo} et~al.}{2021}]{AddVI25}
\begin{barticle}
\bauthor{\bsnm{{Bicudo}}, \binits{P.}},
\bauthor{\bsnm{{Cardoso}}, \binits{N.}},
\bauthor{\bsnm{{Sharifian}}, \binits{A.}}:
\batitle{{Spectrum of very excited $\Sigma_g^+$ flux tubes in SU(3) gauge
  theory}}.
\bjtitle{Phys. Rev. D}
\bvolume{104},
\bfpage{054512}
(\byear{2021})
\doiurl{10.1103/PhysRevD.104.054512}
\end{barticle}
\endbibitem

\bibitem[\protect\citeauthoryear{{Kraus}}{1971}]{Kraus1971GeneralStateChanges}
\begin{barticle}
\bauthor{\bsnm{{Kraus}}, \binits{K.}}:
\batitle{{General State Changes in Quantum Theory}}.
\bjtitle{Ann. Phys.}
\bvolume{64},
\bfpage{311}--\blpage{335}
(\byear{1971})
\doiurl{10.1016/0003-4916(71)90108-4}
\end{barticle}
\endbibitem

\bibitem[\protect\citeauthoryear{{Choi}}{1975}]{Choi1975}
\begin{barticle}
\bauthor{\bsnm{{Choi}}, \binits{M.-D.}}:
\batitle{{Completely positive linear maps on complex matrices}}.
\bjtitle{Linear Algebra Appl.}
\bvolume{10},
\bfpage{285}--\blpage{290}
(\byear{1975})
\doiurl{10.1016/0024-3795(75)90075-0}
\end{barticle}
\endbibitem

\newpage

\bibitem[\protect\citeauthoryear{{Watrous}}{2018}]{Watrous2018}
\begin{bbook}
\bauthor{\bsnm{{Watrous}}, \binits{J.}}:
\bbtitle{The {Theory} of {Quantum} {Information}}.
\bpublisher{Cambridge University Press},
\blocation{Cambridge}
(\byear{2018}).
\doiurl{10.1017/9781316848142}
\end{bbook}
\endbibitem

\bibitem[\protect\citeauthoryear{{Fekete}}{1923}]{Fekete1923}
\begin{barticle}
\bauthor{\bsnm{{Fekete}}, \binits{M.}}:
\batitle{{{\"U}ber die Verteilung der Wurzeln bei gewissen algebraischen
  Gleichungen mit ganzzahligen Koeffizienten}}.
\bjtitle{Math. Z.}
\bvolume{17},
\bfpage{228}--\blpage{249}
(\byear{1923})
\doiurl{10.1007/BF01504345}
\end{barticle}
\endbibitem

\end{thebibliography}
\end{document}


\maketitle
\section{Finite-space conventions and qualification}
\label{suppsec:conventions}

The calculations use the Hamiltonian and site-singlet cutoff of
the article, with $a=1$. The common magnetic constant is removed
and each forward plaquette and its adjoint enter once.
Open electric sums count $N+1$ vertical links and $N$ rail pairs,
without an extra final endpoint rail pair.
Periodic magnetic action includes the wrap pair.
Open paths run between canonical endpoints; periodic paths have
labelled positions. Equal dimensions do not identify the two bases.
Full-sector spectral ordering is the sorted union over any exact
invariant momentum blocks, with multiplicity retained.

For a state vector, the direct full-Hamiltonian residual is
\begin{equation}
\label{supp-eq:input-residual}
 r=\frac{\|H\psi-E\psi\|}{\max(1,|E|)\|\psi\|},
 \qquad E=\frac{\langle\psi,H\psi\rangle}{\langle\psi,\psi\rangle}.
\end{equation}
It is formed from $H\psi-E\psi$, including the periodic wrap
where applicable. A small residual alone does not establish
ground-state ordering or control representation truncation.

\section{Inputs and coherence controls}
\label{suppsec:channel-controls}

The control and resource examples use $B=4$, $q=1$, $N=8$,
$g^2=a=1$ and zero source.
Open columns $3,4,5$ become periodic columns $0,1,2$, with
zero-based indexing. The matched electric terms are
$(V_3,V_4,V_5,U_3,D_3,U_4,D_4)$ and
$(V_0,V_1,V_2,U_0,D_0,U_1,D_1)$, respectively;
the magnetic Hamiltonian terms $(M_3,M_4)$ map to $(M_0,M_1)$.
Operator checks cover all old-leg copies and leakage, not only
selected expectation values.

The input density matrices used in the calculations are represented as
$\rho=FF^\dagger$; mixture weights are included in the columns of $F$.
No block is separately normalized. G is the original zero-source
vector, normalized by its total norm and qualified as a full-sector
ground reference. Its direct relative residual is about
$8.26\times10^{-14}$.
The periodic ground reference is an independent full-sector
minimum; no target eigenvector is used in the controls.

Arms A, B and C were evaluated separately, with the definitions
in the article.
The original-input and own-input distances, and the preprocessing,
map and total energy changes, are defined in Sec.~\ArticleControlSection{} of the article.
Arithmetic uses complex128/float64.
Absolute tolerances are $10^{-10}$ for norm, reference-energy
agreement and energy accounts, $10^{-8}$ for the relative input
residual, $10^{-11}$ for operator identities, and $10^{-12}$ for
channel identities and positivity roundoff.
Block checks use entrywise $10^{-12}$, trace norm
$10^{-12}+10^{-12}p_b$, aggregate half-trace norm
$10^{-11}+10^{-12}$, and weight/leakage $10^{-12}$.
Negative numerical eigenvalues are retained, not clipped.
These checks are diagnostics, not uncertainty intervals.

\section{Fixed resource selection}
\label{suppsec:resource-methods}

The complete bridge witnesses use the cut order and periodic
embedding specified in Sec.~\ArticleResourceSection{} of the article. Each $T_b$ includes
the full-Hamiltonian terms assigned to the bridge; retained,
interface and cut-changing terms are separated as in Eq.~(\ArticleFactorizationEquation)
and Appendix~\ArticleResourceAppendix{} of the article.

The canonical direction is the fixed bridge through $c_1$.
The best basis path minimizes the real diagonal of $T_b$;
exact ties use the lowest coordinate.
For a coherent resource, $T_b$ was checked for Hermiticity and
symmetrized before selecting the lowest numerical cluster of width
$10^{-12}+10^{-12}\|T_b\|_\infty$. The first basis vector whose
projection has norm above $10^{-10}$ was projected onto this
cluster and normalized, with a largest-magnitude component made
real and nonnegative. These fixed rays define the plotted comparison;
every branch is fixed before the six inputs are evaluated.

The numerical one-dimensional clusters do not certify exact
nondegeneracy or interval-optimal energies.
Gains subtract designed output from canonical output;
map costs subtract the unchanged open input, and
$G_{\rm var}$ subtracts the independent periodic ground reference.

\section{Matched spectra}
\label{supp:matched-spectra}
\label{supp:evidence-access}

The twelve unrounded full-sector ground energies correspond to
$N=4,6,8$, sector $q=0,1$ and boundary, at $B=4$ and $g^2=a=1$.
Compute periodic minus open energy at each size, then subtract
the vacuum shift from the flux shift and divide by $Na$.
The differences use these values before plot rounding;
decimal processing adds no physical precision.
The $N=4$ point remains a valid spectral comparison below the
sufficient map threshold; none of the twelve energies is a
reconstructed-state expectation.